\documentclass[12pt]{article}
\usepackage{amsthm,amsmath,amsfonts,amssymb,cases}
\usepackage{graphicx}
\usepackage{xcolor}
\usepackage{cancel}
\usepackage{cite}
\newcommand{\R}{{\mathord{\mathbb R}}}
\newcommand{\Z}{{\mathord{\mathbb Z}}}
\newcommand{\N}{{\mathord{\mathbb N}}}
\newcommand{\C}{{\mathord{\mathbb C}}}

\def\chib {\overline{\chi}}

\newcommand{\KK}{\mathcal{K}}
\newcommand{\HH}{\mathcal{H}}

\newcommand{\FF}{\mathcal{F}}

\newcommand{\VV}{\mathcal{V}}
\newcommand{\WW}{\mathcal{W}}
\newcommand{\hh}{\mathfrak{h}}
\newcommand{\UU}{\mathcal{U}}

\newcommand{\umm}{\underline{m}}
\newcommand{\unn}{\underline{n}}
\newcommand{\upp}{\underline{p}}
\newcommand{\uqq}{\underline{q}}
\newcommand{\uzz}{\underline{0}}

\newcommand{\ran}{{\rm Ran}}

\newcommand{\ben}{\begin{displaymath}}
\newcommand{\een}{\end{displaymath}}
\newcommand{\beqn}{\begin{equation}}
\newcommand{\eeqn}{\end{equation}}
\newcommand{\beqna}{\begin{eqnarray*}}
\newcommand{\eeqna}{\end{eqnarray*}}

\newcommand{\inn}[1]{\langle {#1} \rangle }

\newtheorem{lemma}{Lemma}
\newtheorem{theorem}[lemma]{Theorem}
\newtheorem{remark}[lemma]{Remark}
\newtheorem{proposition}[lemma]{Proposition}

\newtheorem{definition}[lemma]{Definition}

\numberwithin{equation}{section}
\numberwithin{lemma}{section}

\begin{document}
\title{Renormalization Analysis for Degenerate Ground States}
\author{\vspace{5pt} David Hasler and Markus Lange \\
\vspace{-4pt} \small{ Department of Mathematics,
Friedrich Schiller  University  Jena} \\ \small{Jena, Germany } }
\date{}
\maketitle

\begin{abstract}
We consider a Hamilton operator which describes a finite dimensional quantum mechanical system
 with degenerate eigenvalues  coupled to  a field of
relativistic bosons. We show that the ground state projection  and the ground state energy
are analytic functions of the
coupling constant in a cone with apex at the origin, provided a mild infrared assumption holds.
To show the result operator theoretic renormalization is used and extended to degenerate situations.
\end{abstract}

\section{Introduction}

Models of quantum field theory which describe low energy phenomena of quantum
mechanical matter interacting with a quantized field of massless particles have been
 mathematically intensively investigated
(see for example \cite{Spo04} and references therein,  for an early work see \cite{Fro73}).
These  models are used to study
non relativistic matter interacting  with the quantized radiation field or
electrons in a solid interacting with a field of phonons.
Physical properties such as existence of ground states, dispersion relations, and  resonances
 have been treated mathematically rigorous.
In particular, the method of operator theoretic
renormalization, introduced by Bach Fr\"ohlich and Sigal \cite{BacFroSig98-1,BacFroSig98-2},
 has been used in the literature to study    ground states and resonances
\cite{BCFS,GriHas09,Sig09,HasHer11-1,HasHer11-2,
HasHer11-3,FauFroSch14,Fau08,BacCheFroSig07,Che08}.
However, the application of operator theoretic renormalization
usually requires that the unperturbed eigenstate is
non degenerate or at least protected by a symmetry.
In this paper we
 extend  operator theoretic renormalization to situations
where the unperturbed eigenvalue is degenerate and the degeneracy
is lifted after the interaction is turned on.
We note that degenerate situations do occur in physically
realistic models, see for example \cite{AmoFau13}.
To keep  notation  simple we treat the ground state.
Resonances can be treated by the same ideas as used in this
paper with additional notational complexity.
This is planned to be  addressed in a forthcoming paper.

More precisely,
we consider a quantum mechanical  atomic system described by a Hamilton operator acting
on a so called atomic Hilbert space. For simplicity we assume that the
atomic Hilbert space is finite dimensional (we expect that this assumption is not essential and can be relaxed in a straight forward way).
Furthermore, we assume that the  atomic system interacts
with a quantized field of massless bosons by means of a linear coupling. The resulting Hamiltonian
describing the total system is also referred to as generalized spin boson
Hamiltonian. We assume that the interaction  satisfies a mild infrared condition.
The infrared condition is needed for the renormalization analysis to converge. 
It can be shown to include realistic   models of non relativistic quantum electrodynamics  by
means of a so called generalized Pauli Fierz transformation \cite{Sig09}.
We assume that the Hamiltonian  of the atomic subsystem has a degenerate ground state, which
is lifted by formal second order perturbation theory in the coupling constant
(first order perturbation theory does not affect the ground state energy for models which we consider).
We show that the ground
state exists for small values of the coupling constant, a result already known in the literature \cite{Ger00,GriLieLos01,LieLos03,Spo89}. Furthermore,
we show that the ground state projection as well as the ground state energy
are analytic as a function of the coupling constant in an open cone
with apex at the origin. This result is new and it is in  contrast to non degenerate situations,
where it has been shown that the ground state projection and the ground state energy are analytic functions
of the coupling constant \cite{GriHas09}. We do not assume that this is an artefact
of our proof. In fact, we conjecture that  in the degenerate case there may be situations in which the
 ground state projection and possibly the
ground state energy are not analytic
in a neighborhood of zero.
In a related model, where a  hydrogen atom is minimally coupled to the quantized electromagnetic field, non analyticity in the fine structure constant has been
shown \cite{BarCheVouVug10}.

Although we do not obtain analyticity in a neighborhood of zero, analyticity in a cone is of interest in its own right.
It is for example a necessary ingredient to show  Borel summability.  Borel summability
methods allow to recover a function from its asymptotic expansion.
An asymptotic expansion may for example be obtained
using the techniques  employed in   \cite{Ara14,BraHasLan16,HaiSei02,BacFroPiz09,BacFroPiz06}.

In the following section we state the model and the main result.
The subsequent sections are devoted to the proof of the main result.

\section{Model and Statement of Results}

We consider the following model.
Let the atomic Hilbert space be modeled by
$$
\mathcal{H}_{\rm at} = \C^N
$$
and equipped with the standard scalar product. Furthermore we equip  $\mathcal{L}(\HH_{\rm at})$ with the operator norm, which we denote by $\| \cdot \|$.
Let the Fock space
$$
\mathcal{F} = \bigoplus_{n=0}^\infty \hh^{\otimes_s n} ,
$$
with  $\hh = L^2(\R^3 \times \Z_2 )$
model the quantized radiation field.  We denote the Fock vacuum by $\Omega$
and the Hilbert space of the total system by
$$
\HH := \HH_{\rm at} \otimes \mathcal{F} .
$$
We assume that $H_{\rm at} \in \mathcal{L}(\HH_{\rm at})$ is self adjoint.
To simplify our notation we define for $(\boldsymbol{k},\lambda) \in \R^3 \times \Z_2$
\begin{align}\label{eq:easyNotation}
  k := (\textbf{k}, \lambda), \quad
  \int dk :=  \sum_{\lambda=1,2} \int d^3\textbf{k}, \quad  \omega(k) := |k| := |\textbf{k}|
\end{align}
and denote by $a^*(k)$ and $a(k)$ the usual creation and annihilation operator
satisfying canonical commutation relations. For  formal  definitions of the annihilation and creation operator
and associated field operators we refer the reader to   \ref{sec:appfielddef}.
For $G \in L^2(\R^3\times \Z_2 ; \mathcal{L}(\HH_{\rm at}))$ we define
$$
a(G) := \int G^*(k) a(k) dk , \quad    a^*(G) := \int G(k) a^*(k) dk
$$
 which are  densely defined closed linear operators in the Hilbert space.
We define the free field operator by
\begin{equation}  \label{eq:fieldenergy}
H_f := \int{ \omega(k) a^*(k) a(k) dk} ,
\end{equation}
which is defined in the sense of forms.
For $g \in \C$  we  shall study the following operator
\begin{equation}  \label{eq:defhamilton}
H_g := H_{\rm at} +    H_f + g  W     ,
\end{equation}
where  the so called interaction is given by
\begin{equation}  \label{eq:fieldop}
W  = a^*(\omega^{-1/2} G) + a(\omega^{-1/2} G)  .
\end{equation}
We note that $W$ is infinitesimally bounded with respect to $H_f$ if $\omega^{-1} G,\, G  \in L^2(\R^3\times \Z_2 ; \mathcal{L}(\HH_{\rm at}))$.
Let $\epsilon_{\rm at}$ denote the ground state of  $H_{\rm at}$,
 and  let $P_{\rm at}$ denote  the projection onto the eigenspace of $H_{\rm at}$ with
 eigenvalue $\epsilon_{\rm at}$, and let $ \overline{P}_{\rm at} := 1 - P_{\rm at}$.
Define
\begin{align} \label{eq:defofzat}
Z_{\rm at} & : =       -   \int_{}  \frac{dk}{\omega(k)}  P_{\rm at}   G^*(k)    \left[ \frac{ P_{\rm at}}{    |k|} +   \frac{ \overline{P}_{\rm at} }{  H_{\rm at}   - \epsilon_{\rm at}  +    |k|    }
  \right]  G(k)   P_{\rm at}  \quad \upharpoonright \ran P_{\rm at}  ,
 \end{align}
which is a selfadjoint mapping on the ground state space of $H_{\rm at}$.
For $r > 0$ we denote  the open disk in the complex plane by
$$
D_r := \{ z \in \C : |z| < r \} .
$$
In order for the renormalization analysis to be applicable we shall need an infrared condition.
For this we define for   $\mu > 0$
\begin{equation} \label{eq:mu}
L^2_\mu(\R^3\times \Z_2 ; \mathcal{L}(\HH_{\rm at})) :=  \{ G : \R^3\times \Z_2  \to \mathcal{L}(\HH_{\rm at}) :  G \text{ measurable }  , \| G \|_\mu < \infty \} ,
\end{equation}
where we defined
\begin{equation} \label{eq:mu0}
\| G \|_{\mu} :=  \int  \left( \frac{1}{|k|^{3 + 2 \mu}} +  1 \right)  \| G(k) \|^2 dk   .
\end{equation}

\begin{theorem}\label{thm:main} Let $\mu > 0$. Suppose  $G \in L^2_\mu(\R^3\times \Z_2 ; \mathcal{L}(\HH_{\rm at})) $
and let $H_g$ be given by   \eqref{eq:defhamilton}.
Let  $\epsilon^{(2)}_{\rm at}$ denote the smallest eigenvalue of  $Z_{\rm at}$. Assume that $\epsilon^{(2)}_{\rm at}$  is  simple.
Let    $ 0 <  \delta_0  <  \pi/2$, and let
$$S_{\delta_0} := \{ z \in \C : \, | {\rm arg}(z)| < \delta_0 \text{ or } | {\rm arg}(-z)| < \delta_0\} . $$
Then there exists  a $g_0 > 0$  such that for all $g \in D_{g_0} \cap  S_{\delta_0}$
the operator $H_g$ has an  eigenvector $\psi_{g}$ and an eigenvalue $E_{g}$ such that
\begin{equation} \label{expansionmain}
E_{g} = \epsilon_{\rm at} + g^2 \epsilon^{(2)}_{\rm at}   + o(|g|^2) .
\end{equation}
The  eigenvalue and eigenprojection  are  continuous on $S_{\delta_0} \cap D_{g_0}$
and analytic in the interior of $S_{\delta_0} \cap D_{g_0}$.
Furthermore for real $g$  the number  $E_{g}$ is the infimum of the spectrum of  $H_g$.
\end{theorem}

\begin{remark} {\rm
Let $P_\Omega$ denote the projection onto the vacuum vector $\Omega$. We note that
\begin{equation} \label{eq:secondord}
Z_{\rm at}  \cong  - (  P_{\rm at}  \otimes P_\Omega )  W (H_0 - \epsilon_{\rm at}   )^{-1}
 W  ( P_{\rm at}  \otimes P_\Omega ) \upharpoonright \ran  P_{\rm at}  \otimes P_\Omega ,
\end{equation}
which is exactly the second order energy correction in formal perturbation
theory. }
\end{remark}

\begin{remark}\label{rem:fung2}{\rm
If $N = \sum_\lambda \int a_\lambda(k)^* a_\lambda(k)  $ is the number operator we have the symmetry
$(-1)^N H_g (-1)^N = H_{-g}$. This implies that eigenvalues do not depend on the sign of $g$.
And if they happen to have an asymptotic expansion it cannot depend on odd powers of $g$.}
\end{remark}

\begin{remark}{\rm
Generically one can assume that either a  degeneracy of an eigenvalue
remains after the interaction is added (which happens if the degeneracy is protected by a symmetry, e.g., spin degeneracy) or
it is lifted at some finite order. In this paper  we assume that the
degeneracy of the ground state is lifted at second order.
We believe that the methods used in this paper are also usefull
to treat   degeneracies which are lifted at  higher than second  order,
by possibly inserting several initital Feshbach maps, with energy cutoffs
depending on the coupling constant. }
\end{remark}

\begin{remark}{\rm Borel summability methods  allow in certain situations
to recover  a  function  from its   asymptotic expansion, provided it satisfies a strong asymptotic condition.
Theorem \ref{thm:main} together with Remark \ref{rem:fung2} can be used to  show  that the ground state energy as a function of $g^2$ satisfies  the analyticity requirement of  a strong asymptotic condition \cite{ReeSim4}.
Suppose $\pi/4 <  \delta_0  < \pi/2$ and  $g_0$ are as in Theorem \ref{thm:main}.  Define
the function  $f(w) := E_{\sqrt{w}}$ with $| {\rm arg}(w)| < 2 \delta_0$
and $|\sqrt{w}| < g_0$.
Then $f$ is analytic in the interior of the cone
$S_{2\delta_0} \cap D_{g_0^2}$ and extends continuously onto the boundary.
Moreover $f$ satisfies the analyticity requirement for a strong asymptotic
condition. Now suppose there were  $C$ and $\sigma$ such  that
\begin{align}\label{eq:strongasympcond}
	\left| f(w) - \sum_{n = 0}^N c_n w^n \right|
		\leq C \sigma^{N+1} (N+1)! |w|^{N+1}
\end{align}
for all $N$ and all $w \in S_{2 \delta_0} \cap D_{g_0^2}$.
Then $E_g = f (g^2)$ could be recovered uniquely
by the method of Borel summability, see \cite[Watson's theorem]{ReeSim4}.
}
\end{remark}

The remaining part of the paper is devoted to the proof of Theorem \ref{thm:main}.
First we give an overview of the proof, before we comment on the organization of the paper.
As already mentioned in the introduction, the proof is based on an operator theoretic renormalization analysis.
Such an analysis is based  on the so called smooth Feshbach map and its isospectrality properties.
Let us  introduce the relevant  definitions  and properties. For additional details we refer the reader to  \cite{BCFS,GriHas09}.

Suppose  $\chi$ and $\overline{\chi}$ are commuting, nonzero bounded operators, acting on a separable Hilbert space $\KK$ and satisfying
$\chi^2 + \overline{\chi}^2 = 1$.
We shall refer to $\chi$ and $\overline{\chi}$ as smoothed projections.
By a Feshbach pair $(H,T)$ for $\chi$ we mean a pair of closed operators in $\KK$ with the same domain such that the following properties hold:
\begin{itemize}
\item[(i)] $\chi$ and $\overline{\chi}$ commute with $T$,
\item[(ii)] $T$, $H_{\overline{\chi}} := T + \overline{\chi}W \overline{\chi} : D(T) \cap \ran \overline{\chi} \to \ran \overline{\chi}$ are bijections with bounded inverse.
\item[(iii)] $\overline{\chi} H_{\overline{\chi}}^{-1} \overline{\chi} W \chi : D(T) \to \KK$  is a bounded operator.
\end{itemize}
Given a Feshbach pair $(H,T)$ for $\chi$, we  call the operator
$$
F_\chi(H,T) := H_\chi - \chi W \overline{\chi} H_{\overline{\chi}}^{-1} \overline{\chi} W \chi : D(T) \to \KK
$$
 the Feshbach operator. The mapping $(H,T) \mapsto F_\chi(H,T)$ is called  Feshbach map.
Central for our proof is the  isospectrality property  of  the Feshbach map, stated  in the following theorem.
To formulate it one  introduces the so called auxiliary operator
  $$Q_\chi := \chi - \overline{\chi} H_{\overline{\chi}}^{-1} \overline{\chi} W \chi . $$

\begin{theorem}[\cite{GriHas09}] Let $(H,T)$ be a Feshbach pair for $\chi$ on a separable Hilbert space $\KK$.  Then
\begin{align*}
& \chi : \ker H \to \ker F_\chi(H,T) \\
& Q_\chi   :  \ker F_\chi(H,T)  \to  \ker H
\end{align*}
are linear isomorphisms and inverse to each other.
\end{theorem}

The renormalization analysis is based on a  carefully designed  iterated application of the smooth Feshbach map.
Before one can start this machinery, we need to control the degeneracy.
For this, we   perform   two so called initial Feshbach maps.
The first Feshbach map uses a smoothed projection given as a tensor product. The first factor acts on the
atomic Hilbert space and  projects orthogonally onto  $1_{H_{\rm at}=\epsilon_{\rm at}}\HH_{\rm at}$, the degenerate ground state eigenspace of the atomic Hamiltonian.
The second factor acts  on Fock space and is given by a smooth cutoff function of the free field energy,
such that its   range is contained  in a  spectral subspace of field energies between zero and $\rho_0$. 
 The second Feshbach map uses a smoothed projection given again as a tensor product. The first
factor now projects onto the  subspace of $1_{H_{\rm at}=\epsilon_{\rm at}}\HH_{\rm at}$ belonging to the lowest energy eigenvalue  of $Z_{\rm at}$.
The second factor is given again by a smooth cutoff function of the free field energy,  with range contained
in a  spectral subspace of  even smaller field energies between zero and $\rho_0 \rho_1$.
 This procedure  will   resolve the  degeneracy and hence  allows us to initiate the usual renormalization analysis in Fock space
as introduced in   \cite{BCFS}.

In the proof we will choose  $\rho_0$   larger than  $|g|$
but  $\rho_0 \rho_1$  smaller than $|g|^2$. More precisely, we will  show  the following.
 For any   $\rho_0 > 0$ there exists a  $\rho_1 > 0$ and
  positive numbers $ g_-(\rho_0), \, g_+(\rho_0)$ with $g_-(\rho_0) < g_+(\rho_0)$,
 uniformly in the model parameters,  such that
 for all coupling constants $g$  in a sectorial region of the complex plane with
 \begin{equation} \label{eq:basicanaest}
 g_-(\rho_0)  <   |g| <  g_+(\rho_0)
 \end{equation}
 both  Feshbach maps are isospectral and respect necessary  Banach space estimates needed for operator theoretic  renormalization to be applicable, see Figure~\ref{pic:wedgeofanalyticity}.
 Then for $g$ in a sectorial region  of an annulus we shall obtain  analyticity of the ground state projection and
ground state energy by invoking  an analyticity result of Griesemer and Hasler \cite{GriHas09}.
Moreover, we show that we can choose    $g_-(\rho_0)$ such that $g_-(\rho_0) \to 0$ as $\rho_0 \to 0$ (at the same time $g_+(\rho_0) \to 0$ but this
is no problem as long as  $g_-(\rho_0) < g_+(\rho_0)$).
Hence the analyticity in a cone with apex at the origin   will follow as  $\rho_0$ tends to zero.

\begin{figure}[ht]
	\centering
 \includegraphics[width=1.0\linewidth]{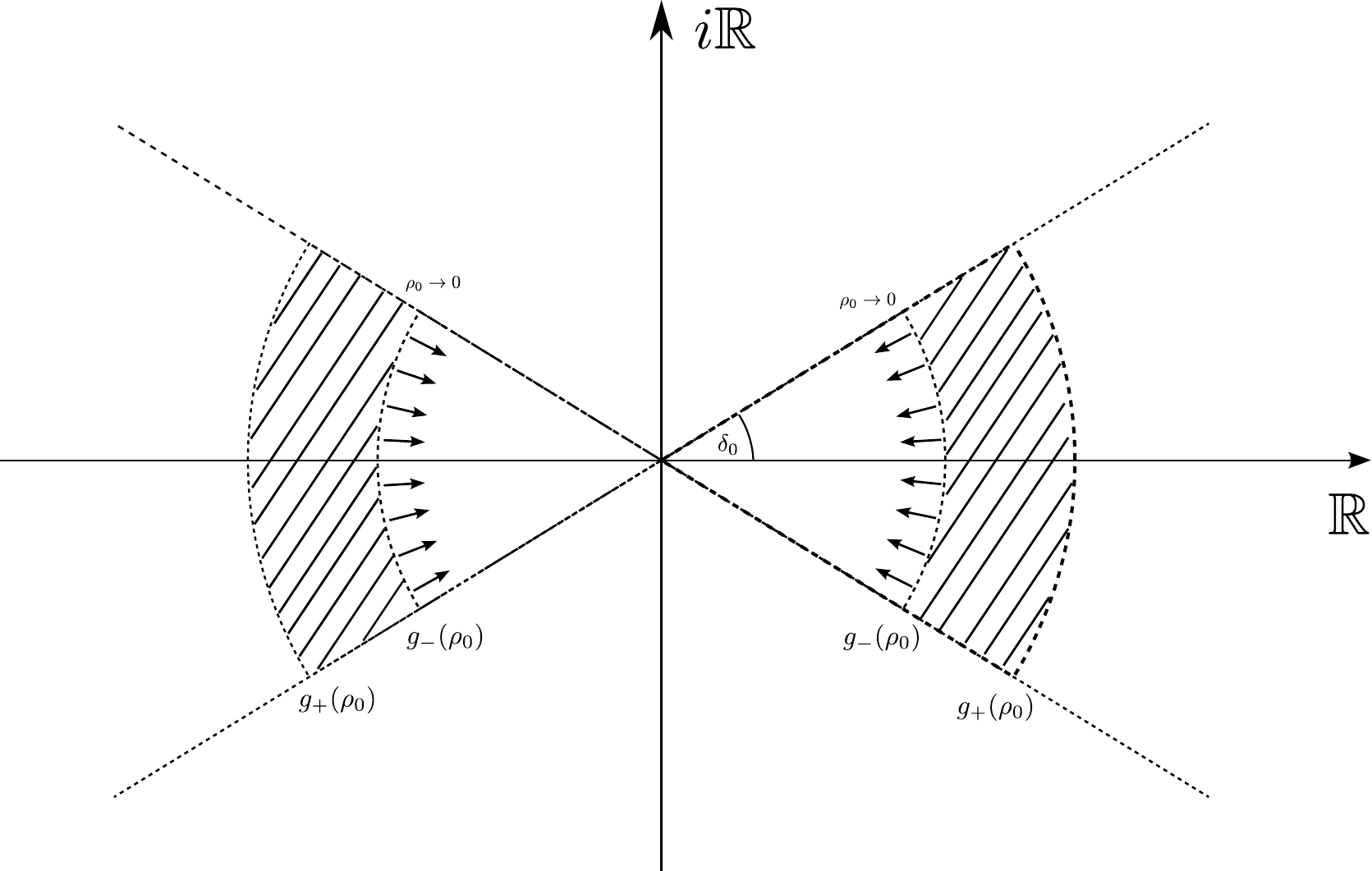}
	\caption{For fixed $\rho_0>0$, the ground state projection and the ground state energy are analytic functions of the coupling constant in the shaded region.}
	\label{pic:wedgeofanalyticity}
\end{figure}

The  remaining part of the paper is organized as follows.
In Section~\ref{sec:firstfeshstep} we define  the first   Feshbach map with parameter $\rho_0$  and prove the
Feshbach pair criterion, which establishes isospectrality.
In Section~\ref{sec:Banachspace} we define  Banach spaces of matrix valued integral kernels. These Banach spaces
  parametrize   subspaces of operators in Fock space  and allow to control the flow of the  operator theoretic renormalization analysis.
 Thereby we extend the Banach spaces of   integral kernels introduced in \cite{BacCheFroSig07}  to
 include matrix valued integral kernels,   in order  to adapt
 degenerate situations.
 In Section \ref{sec:firstStep} we  introduce polydiscs  around the free field energy in a subspace parametrized by a
  Banach space of integral kernels, \eqref{eq:defofBBals}.
 We then show in Theorem~\ref{thm:inimain1} that the first Feshbach operator lies in such a   polydisc, where the radii of the polydisc
  can be made arbitrarily small  for sufficiently small $|g|$.
In Section~\ref{sec:secondFeshStep}  we define  the second Feshbach map with parameter $\rho_1$, which
will be responsible for lifting the degeneracy.
We establish necessary estimates, which will
be used  to show the  Feshbach pair criterion for the second step.
We note that for these estimates, in particular Lemma~\ref{T-invert},  we need  the  coupling constant to lie in  a cone.
We conclude the section with an abstract Feshbach pair criterion.
In Section~\ref{sec:secondStep} we prove an abstract Banach space estimate for the
second Feshbach step, which will be used to show that the second Feshbach operator lies
in a suitable polydisc around the free field energy.
Section~\ref{sec:prooofmain} is devoted to the proof of the main result  Theorem~\ref{thm:main}.
Having the abstract
results of the previous two sections at hand, we  will finally  choose $\rho_1$ as a function of $\rho_0$,
such that the second Feshbach map  satisfies indeed the Feshbach
pair criterion, establishing  isospectrality, and such that the second Feshbach operator
lies in a suitable polydisc, allowing the application
of the analyticity result  \cite{GriHas09}. In particular, we rigorously justify (see \eqref{eq:proofofmain1})
the situation  outlined  in Figure~\ref{pic:wedgeofanalyticity} together with   \eqref{eq:basicanaest}.

\section{Initial Feshbach Step}
\label{sec:firstfeshstep}

Without loss we  assume that the distance of $\epsilon_{\rm at}$ from the rest of the spectrum of $H_{\rm at}$ is
 $1$, i.e.,
\begin{equation} \label{eq:condonenedist}
d_{\rm at} := \inf \sigma(H_{\rm at} \setminus \{ \epsilon_{\rm at} \} ) - \epsilon_{\rm at}  = 1  .
\end{equation}
  This can always be achieved by a suitable scaling.
In this section we define the  initial Feshbach operator. For definitions and conventions we refer  to \cite{BCFS,GriHas09}
and collect for the convenience of the reader  basic notions in the appendix.
We will show  for $z$ in a neighborhood of $\epsilon_{\rm at}$ and for
$|g|$ sufficiently small,  that
$H_g - z$ and $H_0 - z$ are  a Feshbach pair for  a generalized projection, which we will now introduce.
We   fix two  functions  $\chi$ and $\overline{\chi}$  in  $C^\infty(\R ; [0,1])$ satisfying the
following properties:
\begin{itemize}
\item[(i)]  $\chi(t) = 1$, if $t \leq 3/4$, and  $\chi(t) = 0$, if $t \geq 1$,
\item[(ii)]   $\chi^2 + \overline{\chi}^2 = 1$.
\end{itemize}
For $\rho > 0$ we define the following operator valued functions
\begin{align*}
\boldsymbol{\chi}_\rho^{(0)}(r) & :=   P_{\rm at} \otimes \chi(r/\rho)  \\
\overline{\boldsymbol{\chi}}_\rho^{(0)}(r) & :=  \overline{P}_{\rm at} \otimes 1 + P_{\rm at} \otimes \overline{\chi}(r /\rho)
\end{align*}
and we define by means of the spectral theorem the linear operators in $\mathcal{L}(\HH)$
\begin{align*}
\boldsymbol{\chi}_{\rho}^{(0)} & :=    \boldsymbol{\chi}_{\rho}^{(0)}(H_f)     \\
\overline{\boldsymbol{\chi}}_\rho^{(0)} & :=  \overline{\boldsymbol{\chi}}_\rho^{(0)}(H_f)  .
\end{align*}
It is easily verified that   ${(\boldsymbol{\chi}_\rho^{(0)})}^2 + {(\overline{\boldsymbol{\chi}}_\rho^{(0)})}^2  = 1 $.

The following proposition provides us with conditions for which the initial Feshbach operator will be  defined.

\begin{theorem}[Feshbach Pair Criterion  for 1st Iteration]
 \label{pro:firstbound}
Let $\omega^{-1/2} G \!\in\! L^2(\R^3\times \Z_2 ;\! \mathcal{L}(\HH_{\rm at})\!)$,
$0 < \rho \leq \frac{1}{4}$ and $z \in D_{\rho/2}(\epsilon_{\rm at})$.
The operators $H_g-z$ and $H_0-z$ are a Feshbach pair for $\boldsymbol{\chi}_{\rho}^{(0)}$,
if
\begin{align} \label{eq:upperboundong}
	|g| < \frac{ \rho^{1/2}}{10     \|\omega^{-1}G \| } .
\end{align}
Furthermore
one has the absolutely convergent expansion
\begin{equation} \label{eq:firstresolventexp}
F_{\boldsymbol{\chi}_{\rho_0}^{(0)}}( H_g -z   ,   H_0 - z )
	= \epsilon_{\rm at} - z  + H_f   +  \sum_{L=1}^\infty(-1)^{L-1} \boldsymbol{\chi}_{\rho}^{(0)} g W \left(  \frac{ {( \overline{\boldsymbol{\chi}}_{\rho}^{(0)})}^2 }{H_0 - z}  g W  \right)^{L-1}   \boldsymbol{\chi}_{\rho}^{(0)}.
\end{equation}
\end{theorem}

Let us now provide proof of  Theorem   \ref{pro:firstbound}.
For this we shall make use of  the following two lemmas.

\begin{lemma}\label{lem:estonint}  Let $\rho \geq  0$. Then for  $G \in L^2(\R^3 \times \Z_2 ; \mathcal{L}(\HH_{\rm at}))$ we have
$$
\| (H_f + \rho )^{-1/2} [ a( G) + a^*( G) ] (H_f + \rho)^{-1/2} \|    \leq  2 \| \omega^{-1/2} G\|   \rho^{-1/2}
$$
\end{lemma}
\begin{proof} This follows from  \eqref{eq:estoncrea}.
For the annihilation  operator we find for $r \geq 0$
\begin{align*}
	\| (H_f+r)^{-1/2} a(G) &(H_f + r )^{-1/2} \| \\
		& \leq  \| (H_f+r)^{-1/2} \| \| a(G) H_f^{-1/2} \| \| H_f^{1/2} (H_f + r )^{-1/2} \| \\
		& \leq   r^{-1/2} \| \omega^{-1/2}  G \|  ,
\end{align*}
and likewise we estimate the corresponding expression involving a creation operator.
\end{proof}

\begin{lemma} \label{lem:firstfreeresolvbound}
Let $0 < \rho \leq \frac{1}{4}$ and  $z \in D_{\rho/2}(\epsilon_{\rm at})$.
The operator  $H_0-z$ is invertible on the range of $\overline{\boldsymbol{\chi}}_{\rho}^{(0)}$
and we have the bound
\begin{align} \label{eq:firstfreeresolvbound1}
\| ( H_0 - z )^{-1}   \upharpoonright \ran   \overline{\boldsymbol{\chi}}_{\rho}^{(0)}   \|
&\leq     \frac{4}{ \rho}
\end{align}
and for  all $\tau \geq 0$  the bound
\begin{align} \label{eq:firstfreeresolvbound2}
\|(H_f + \tau)^{1/2}  ( H_0 - z )^{-1} (H_f + \tau)^{1/2}  \upharpoonright \ran   \overline{\boldsymbol{\chi}}_{\rho}^{(0)} \|
&\leq 1 +  \frac{4 \tau}{ \rho}  .
\end{align}
\end{lemma}

\begin{proof}
First we show that $H_0 - z $ is bounded invertible on the range of $\overline{\boldsymbol{\chi}}_{\rho}^{(0)}$. It follows
 for a normalized $\psi \in \ran ( {P}_{\rm at}  \otimes \chib(H_f/\rho) ) $ that
\begin{equation*}
\| ( H_0 - z ) \psi \| \geq
\inf_{r \geq \frac{3}{4} \rho  } \left| \epsilon_{\rm at} + r - z \right|
\geq  (3/4-1/2)\rho = \frac{1}{4} \rho ,
\end{equation*}
and it follows  from \eqref{eq:condonenedist} that  for a normalized $\psi \in  \ran  ( \overline{P}_{\rm at}  \otimes 1 )  $ we have
\begin{align*}
\| ( H_0 - z ) \psi   \| &\geq
\inf_{r \geq  0} \| ( H_{\rm at}\overline{P}_{\rm at} + r - z  ) \psi  \|  \\
	&\geq  \inf_{r \geq  0 }\| ( H_{\rm at}\overline{P}_{\rm at}- \epsilon_{\rm at} +  r ) \psi  \|
		- |z-\epsilon_{\rm at}| \\
	&\geq  1 - \rho/2 .
\end{align*}

Thus from the above two inequalities it follows   that  $H_0 - z $ is bounded invertible on the range of $\overline{\boldsymbol{\chi}}_{\rho}^{(0)}$
and that   \eqref{eq:firstfreeresolvbound1} holds.
We estimate  further,
\begin{align*}
\| (H_f + \tau)^{1/2}  ( H_0 - z )^{-1} &(H_f + \tau)^{1/2}
	\upharpoonright \ran ( {P}_{\rm at}  \otimes \chib(H_f/\rho)  )  \|   \\
& =
\sup_{r \geq \frac{3}{4} \rho  } \left|
\frac{r + \tau  }{ \epsilon_{\rm at} + r - z }
\right|
\leq 1 + \left| \frac{\tau}{(3/4 - 1/2)\rho} \right|
\leq 1 +  \frac{4\tau}{\rho}
\end{align*}
and if we set $E_1 := \inf \sigma(H_{\rm at})\setminus \{\epsilon_{\rm at}\}$
\begin{align*}
\|(H_f + \tau)^{1/2}  ( H_0 - z )^{-1} &(H_f + \tau)^{1/2} \upharpoonright \ran  ( \overline{P}_{\rm at}\otimes 1 )   \|  \\
& =
 \sup_{r \geq  0} \sup_{\lambda \in \sigma(H_{\rm at})\setminus\{\epsilon_{\rm at}\}} \left|
\frac{ r + \tau  }{ \lambda + r - z } \right| \\
& \leq
 \sup_{r \geq  0} \left|
\frac{ r + \tau }{ E_{1} - \epsilon_{\rm at} - \rho/2 + r}  \right| \\
&\leq  1 +   \frac{ \tau }{1 - \rho/2} .
\end{align*}
Now the above two inequalities  show  \eqref{eq:firstfreeresolvbound2}.
\end{proof}

\vspace{0.5cm}
\noindent
{\it Proof of Theorem \ref{pro:firstbound}.}
To prove the invertibility, we note that $H_f$ and $H_0$ leave the range of $\overline{\boldsymbol{\chi}}_{\rho}^{(0)}$
invariant and  that   $(H_f + \tau)^{1/2}$ is bounded invertible on
the range of  $\overline{\boldsymbol{\chi}}_{\rho}^{(0)}$.
We use Lemmas  \ref{lem:firstfreeresolvbound} and \ref{lem:estonint}.
We write
\begin{align*}
A(z,\tau) & :=   ( H_f + \tau)^{-1/2}     (   H_0 - z )  (H_f+ \tau)^{-1/2}  \\
B(z,\tau,\rho) & := (H_f+ \tau)^{-1/2} \overline{\boldsymbol{\chi}}_{\rho}^{(0)}      W  \overline{\boldsymbol{\chi}}_{\rho}^{(0)}   (H_f + \tau)^{-1/2}
\end{align*}
and we use the following identity
\begin{align*}
(  & H_0 - z + g     \overline{\boldsymbol{\chi}}_{\rho}^{(0)}    W  \overline{\boldsymbol{\chi}}_{\rho}^{(0)}  )  \upharpoonright \ran \overline{\boldsymbol{\chi}}_{\rho}^{(0)}    \\
  & =
   (H_f+\tau)^{1/2} [
    A(z,\tau)
      +
    g B(z,\tau,\rho)    ]    ( H_f + \tau)^{1/2}    \upharpoonright \ran \overline{\boldsymbol{\chi}}_{\rho}^{(0)}  \\
    & =
   (H_f+\tau)^{1/2}
    A(z,\tau) [ 1
      +
    g A(z,\tau)^{-1}   B(z,\tau,\rho)    ]   ( H_f + \tau)^{1/2}    \upharpoonright \ran \overline{\boldsymbol{\chi}}_{\rho}^{(0)} .
\end{align*}
Thus the bounded invertibility follows from Neumanns theorem provided
$$
 \|  g A(z,\tau)^{-1}   B(z,\tau,\rho)   \|  < 1.
 $$
Now using   Lemma \ref{lem:estonint}, note that $\overline{\boldsymbol{\chi}}_{\rho}^{(0)}$  commutes with $H_f$,  and  \eqref{eq:firstfreeresolvbound2} of Lemma \ref{lem:firstfreeresolvbound},  we obtain
$$
 \|  g A(z,\tau)^{-1}   B(z,\tau,\rho,)   \|   \leq  | g | \left(  1 +  \frac{4 \tau}{ \rho} \right)  2 \| \omega^{-1} G \|  \tau^{-1/2} .
$$
Thus, if we choose $\tau = \rho$, then the proposition follows.
\qed

\section{Banach Space of Integral Kernels}
\label{sec:Banachspace}
To control  the renormalization transformation, in particular proving its convergence,
we introduce the following Banach spaces of integral kernels. We note that the choice
of these spaces is not unique. We choose the Banach spaces such that they are suitable for
the model which we consider. Specifically, the Banach spaces which we introduce below are a straight forward
generalization of the spaces defined  in  \cite{BCFS} or \cite{GriHas09} to matrix
valued integral kernels. A generalization to matrix valued integral kernels seems to be a canonical
choice to accommodate  degenerate situations.

For  $d  \in \N$ we define  the Banach space $\mathcal{W}_{0,0}^{[d]}$ as  the space of continuously
differentiable matrix valued functions
$$
\mathcal{W}_{0,0}^{[d]} := C^1([0,1]; \mathcal{L}(\C^{d})  )
$$
with norm
$$
\| w \|_{C^1}  := \| w \|_\infty + \| w' \|_\infty ,
$$
where $(\cdot)'$ stands for the derivative.
Let $\boldsymbol{B}_1 := \{ \boldsymbol{k} \in \R^3 : |\boldsymbol{k}| \leq 1 \}$.
For a set $\boldsymbol{A} \subset \R^3$ we write
$$
A := \boldsymbol{A} \times \{1,2\}  , \quad
\int_{A} dk := \sum_{\lambda=1,2} \int_{\boldsymbol{A}} d^3\boldsymbol{k} ,
$$
where we recall  the notation of Eq. \eqref{eq:easyNotation}.
For $m,n \in \N$ we write
\begin{align*}
k^{(m)} &:= (k_1,...,k_m) \in   ( \R^3 \times \{1,2\})^m  , \\
\tilde{{k}}^{(n)} &:= (\tilde{{k}}_1,...,\tilde{{k}}_n) \in  ( \R^3 \times \{1,2\})^n , \\
{K}^{(m,n)} &:= ({k}^{(m)} , \tilde{{k}}^{(n)})
\end{align*}
\begin{align*}
d {k}^{(m)} &:= d {k}_1 \cdots d {k}_m    , \\
d \tilde{{k}}^{(n)} &:= d \tilde{{k}}_1 \cdots d \tilde{{k}}_n  \\
d {K}^{(m,n)} &:= d {k}^{(m)} d  \tilde{{k}}^{(n)}
\end{align*}
\begin{align*}
|{k}^{(m)} | &:= | {k}_1| \cdots |{k}_m | ,  \\
|\tilde{{k}}^{(n)}| &:= |\tilde{{k}}_1| \cdots |\tilde{{k}}_n|  \\
 |{K}^{(m,n)}| &:= |{k}^{(m)} | |\tilde{{k}}^{(n)}| .
\end{align*}
Moreover, we shall use
\begin{align*}
& \Sigma[{k}^{(n)}] := \sum_{i=1}^n |k_i | , \qquad
\Sigma[\tilde{{k}}^{(m)}] := \sum_{i=1}^m |\tilde{k}_i |.
\end{align*}
 For $m,n \in \N$ with $m+n \geq 1$ and $\mu > 0$  let
 $\mathcal{W}_{m,n}^{[d]}$ denote the space of  measurable functions $w_{m,n} : {B}_1^{m+n} \to \mathcal{W}_{0,0}^{[d]}$
 satisfying the following three properties.
 \begin{itemize}
 \item[(i)]  $w_{m,n}$  are symmetric with respect to all permutations of the $m$ arguments
 from ${B}_1^m$ and the $n$ arguments from ${B}_1^n$, respectively;
 \item[(ii)]  for $m + n \geq 1$,  we have
 $ w_{m,n}({k}^{(m)} , \tilde{{k}}^{(n)})(r) = 0$,
 provided $ r \! \geq 1 - \max(\Sigma[{k}^{(m)}], \Sigma[\tilde{{k}}^{(n)}])$;
 \item[(iii)]  The following norm is finite
   $$
 \| w_{m,n} \|_\mu^{\#} := \| w_{m,n} \|_\mu + \| \partial_r  w_{m,n}  \|_\mu ,
 $$
 where
  $$
 \| w_{m,n} \|_\mu := \left( \int_{{B}_1^{m+n}}
  \| w_{m,n}({K}^{(m,n)} ) \|_{\infty}^2
  \frac{d {K}^{(m,n)}}{|{K}^{(m,n)}|^{3 + 2 \mu}}   \right)^{1/2} .
 $$
\end{itemize}

\begin{remark}{\rm
We note that  $\mathcal{W}_{m,n}^{[d]}$ is given as the  subspace of
\begin{equation}\label{eq:defofl22}
 L^2 \left( {B}^{m+n} , \frac{d {K}^{(m,n)}}{|{K}^{(m,n)}|^{3 + 2 \mu} } ; \mathcal{W}_{0,0}^{[d]} \right)  ,
\end{equation}
consisting of elements satisfying  Conditions (i) and (ii)  above.
Note that the norm $\| \cdot \|_\mu^\#$ is equivalent to the natural norm of  \eqref{eq:defofl22} (given by the theory of Banach space valued $L^p$-functions), which is the norm chosen in  \cite{GriHas09}.
Moreover, we will identify \eqref{eq:defofl22} as a subspace of $ L^2 \left([0,1] \times {B}^{m+n}  , \frac{d {K}^{(m,n)}}{|{K}^{(m,n)}|^{3 + 2 \mu} } ; \mathcal{L}(\C^d) \right)$
by means of
\begin{equation} \label{eq:identbsp}
 w_{m,n}(r, {k}^{(m)} , \tilde{{k}}^{(n)}) =  w_{m,n}({k}^{(m)} , \tilde{{k}}^{(n)})(r)  .
\end{equation}
Henceforth we shall use this identification without comment.
Furthermore,  we remark that
$$
	\|w_{0,0}\|_{C^1} = \|w_{0,0}\|_\mu^\#
$$
 with the natural convention that for $m=n=0$ the empty Cartesian product consists of a
 single point and that there  is no  integration in that case.}
 \end{remark}

For given $\xi \in (0,1)$ and $\mu > 0$ we define the Banach space
$$
\WW_\xi^{[d]} := \bigoplus_{m,n \in \N_0} \WW_{m,n}^{[d]}
$$
with norm
$$
\| {w} \|_{\mu, \xi}^\#
	:=  \| w_{0,0} \|_{C^1} +  \sum_{m,n \geq 1} \xi^{-(m+n)} \| w_{m,n} \|_\mu^\# ,
$$
for $w = (w_{m,n})_{m,n\in \N_0}\in \WW_\xi^{[d]}$.

Next we define a linear mapping $H : \WW_\xi^{[d]} \to \mathcal{L}(\HH_{\rm red})$, where $\HH_{\rm red} := P_{\rm red}\HH$ and  $P_{\rm red} := 1_{[0,1]}(H_f)$. For this we will  use the notation
\begin{equation*}
 a^*({k}^{(m)}) :=  \prod_{i=1}^m a^*({k}_i)    \quad ,  \quad \quad
  a(\tilde{{k}}^{(n)}) := \prod_{i=1}^n a(\tilde{{k}}_i) .
\end{equation*}
If $w \in \WW_{0,0}^{[d]}$ define $H_{0,0}(w) := w_{0,0}(H_f)$.
For $m+n \geq 1$ and  $w_{m,n} \in  \mathcal{W}_{m,n}^{[d]}  $
define the operator on $\HH_{\rm red}$
\begin{equation} \label{eq:defofhmn}
H_{m,n}(w_{m,n}) := P_{\rm red} \!\left( \!\int_{B_1^{m+n}}\!\! a^*(k^{(m)}) w_{m,n}(H_f, {K}^{(m,n)} ) a(\tilde{k}^{(n)}) \frac{d  {K}^{(m,n)}}{|{K}^{(m,n)}|^{1/2}} \!\right) \!P_{\rm red}.
\end{equation}
The following lemma is proven in \cite[Theorem 3.1]{BCFS}

\begin{lemma} \label{lem:operatornormestimates}  Let $\mu > 0$ and $m,n \in \N_0$ with
$m+n \geq 1$.
For $w_{m,n} \in \mathcal{W}_{m,n}^{[d]}$ we have
\beqn \label{eq:operatornormestimate1}
\|H_{m,n}(w_{m,n}) \|  \leq  \frac{ \| w_{m,n} \|_{\mu} }{ \sqrt{  n^n m^m } }  \; .
\eeqn
\end{lemma}
Using our notation above and the convention that $p^p := 1$ for $p = 0$ it is
trivial to extend this lemma to the case $m + n = 0$.

For  sequences $w = (w_{m,n})_{(m,n) \in \N_0^2}   \in \WW_\xi^{[d]}$ we define the  operator $H(w)$ as the sum
\begin{align} \label{eq:opHw}
H(w) := \sum_{m,n} H_{m,n}(w_{m,n}) ,
\end{align}
where the sum converges in operator norm, which can be seen using   \eqref{eq:operatornormestimate1}.

A proof of the following theorem can be found in \cite{BCFS} with a modification explained in \cite{HasHer08-2}.

\begin{theorem} \label{thm:injective} Let $\mu > 0$ and $0 < \xi < 1$.
Then the map $H : \WW_\xi^{[d]}  \to \mathcal{B}(\HH_{\rm red})$ is injective and bounded.
For $w \in \WW_\xi^{[d]}$ and for $ \tilde{w} \in \WW_\xi^{[d]}$ with  $\tilde{w}_{0,0} = 0$ we have
\begin{equation}   \label{eq:injective}
\| H(w) \|    \leq \| w \|_{\mu,\xi}^\#  , \quad \|H(\tilde{w}) \|  \leq \xi   \| \tilde{w} \|_{\mu,\xi}^\#.
\end{equation}
\end{theorem}

The renormalization transformation will involve a rescaling of the energy.
This rescaling is described  by means of a dilation operator, which we shall now define.

\begin{definition} Let  $\rho > 0$.  We define the operator of dilation   on the one particle sector  by
$$
U_\rho : \hh \to \hh  ,  \quad (U_\rho \varphi)(k) = \rho^{3/2} \varphi(\rho k )  ,
$$
where we use the notation
$
	\rho k := (\rho \boldsymbol{k}, \lambda) .
$
We define the operator of  dilation on Fockspace by
$$
\Gamma_\rho =  \bigoplus_{n=0}^\infty   U_\rho^{\otimes n}  \upharpoonright \FF .
$$
 We define the mapping  $S_\rho : \mathcal{L}(\FF) \to \mathcal{L}(\FF)$ called
 rescaling by dilation by
\begin{align} \label{eq:scalingTransfo}
S_\rho(A) := \rho^{-1} \Gamma_{\rho}(A ) \Gamma_{\rho}^* .
\end{align}
\end{definition}

\begin{remark} \label{rm:scalingTransfo} {\rm
We note that 
that by definition
$
\Gamma_\rho \Omega = \Omega .
$
Moreover, one can show that
$$
\Gamma_\rho a^*(k) \Gamma_\rho^* = \rho^{-3/2} a^*(\rho^{-1} k)  , \quad
\Gamma_\rho a(k) \Gamma_\rho^* = \rho^{-3/2} a(\rho^{-1} k) .
$$}
\end{remark}

The following lemma relates the scaling transformation to a  scaling transformation of the
integral kernels. It is straight forward to verify by the substitution formula.

\begin{lemma} For  $w  \in \WW_{\xi}^{[d]}$  define the scaling transformation of the integral kernel by
\begin{equation} \label{eq:scaling}
	s_\rho(w_{m,n})[r,K^{(m,n)}] := \rho^{(m+n)-1} w_{m,n}[\rho r, \rho K^{(m,n)}].
\end{equation}
Then
$$
	S_\rho(H(w)) = H(s_\rho(w)) . 
$$
\end{lemma}

\begin{remark} {\rm For  $m+n \geq 1$  one finds
$$
	\|s_\rho(w_{m,n})\|_\mu \leq \rho^{\mu(m+n)}\|w_{m,n}\|_\mu.
$$
This  illustrates that after a renormalization step, which will be introduced below,
the relative size of the
perturbative part $\|w_{m,n}\|_\mu$, $m + n \geq 1$, will shrink compared to the unperturbed
part $\|w_{0,0}\|_{C^1}$.  } 
\end{remark}

\section{Banach Space Estimate for the first  Step}
\label{sec:firstStep}

In Section~\ref{sec:firstfeshstep} we  showed  that the operators  $H_g - z$ ,    $ H_0 - z $    are a Feshbach pair for $\boldsymbol{\chi}_{\rho}^{(0)}$,
provided the  coupling constant  is in
a sufficiently small neighborhood of zero and the  spectral parameter is sufficiently close to the unperturbed ground state energy.
In particular, if   the assumptions of   Theorem   \ref{pro:firstbound} hold,   we can define
 the operator
\begin{align}\label{eq:initialFeshbachOperator}
{H}_g^{(1,\rho)}(z) :=   S_{\rho}(  F_{\boldsymbol{\chi}_{\rho}^{(0)}}( H_g - z  ,    H_0 - z  ) ) ,
\end{align}
which we call the first   Feshbach operator.
The goal of this section is to show  that the first Feshbach operator
is close to the free field energy. The distance  will be measured in terms of the norms introduced in
 the previous section. More precisely, we define the following   polydiscs  of the free field energy.
For  given $\alpha, \beta, \gamma \in \R_+$ we define
$\mathcal{B}^{[d]}(\alpha, \beta, \gamma) \subset H(\WW_\xi^{[d]})$ 
by
\begin{equation}  \label{eq:defofBBals}
\mathcal{B}^{[d]}(\alpha, \beta, \gamma ) := \left\{ H(w) : \|w_{0,0}(0) \| \leq \alpha , \| w_{0,0}' - 1 \|_\infty \leq \beta, \| w - w_{0,0} \|_{\mu , \xi }^\# \leq \gamma \right\} .
\end{equation}
We  shall denote the dimension of the space of ground states of $H_{\rm at}$ by
$$
d_0 := \dim(\ran P_{\rm at}) .
$$
Moreover, we introduce the following  global constant
\begin{equation}  \label{eq:defofCF}
C_F := 10 \| \chi'\|_\infty + 20 ,
\end{equation}
\begin{equation}  \label{eq:defofCFhat}
\hat{C}_F :=  20 ,
\end{equation}
which will be used in various estimates. For the renormalization analysis to be applicable, we need a stronger
infrared condition, which is expressed in terms of the norm $\| \cdot \|_\mu$.
We note that as a consequence of the definition we have
\begin{equation}
 \label{eq:diffnormonGest}
 \left\| \frac{ G}{\omega}  \right\| \leq \| G \|_\mu .
\end{equation}
This inequality  shows that the criterion for  the Feshbach pair property obtained in Section  \ref{sec:firstfeshstep} can be expressed in terms of  $\| G \|_\mu$.
We can now state the main theorem of this section.

\begin{theorem}[Banach Space Estimate for 1st Feshbach Operator] \label{thm:inimain1}
Let  $G \in L^2_\mu(\R^3\times \Z_2 ; \mathcal{L}(\HH_{\rm at}))$, and $\xi \in (0,1)$.
Then there  exist constants  $C_1, C_2, C_3$,
 such that,
if     $0 < \rho < 1/4$,  $z \in D_{\rho/2}(\epsilon_{\rm at})$ and
\begin{equation} \label{eq:condongrho1}
|g| < C_0  \rho^{1/2}    , \qquad \text{where} \quad  C_0 := \frac{1}{ 8 \xi^{-1} C_F  \| G \|_\mu  }   ,
\end{equation}
the pair of operators   $(H_g -z, H_0 - z)$ is a Feshbach pair for $\boldsymbol{\chi}_{\rho}^{(1)}$, and
$$
H_g^{(1,\rho)}(z) - \rho^{-1} (\epsilon_{\rm at}  - z )  \in  \mathcal{B}^{[d_0]}(\alpha_0, \beta_0, \gamma_0 )  ,
$$
for
$$
\alpha_0 =  C_1 |g|^2 \rho^{-1}  , \quad  \beta_0  =  C_2 |g|^2  \rho^{-1}   , \quad \gamma_0 = C_3 \rho^\mu ( |g| + \rho^{-1} |g|^2 + \rho^{-2} |g|^3 ) .
$$
\end{theorem}

\begin{remark} {\rm We note that the explicit form of the  constants $C_1$, $C_2$, and $C_3$ can be read off from  Inequalities
 \eqref{eq:feb2:1}--\eqref{eq:feb2:3}. Only  $C_3$ depends on $\xi$.  }
\end{remark}

The remaining part of this section is devoted to the proof of Theorem \ref{thm:inimain1}.
First observe that  in view of \eqref{eq:diffnormonGest} and   \eqref{eq:defofCF}  we see that assumption  \eqref{eq:condongrho1}   implies 
the assumption \eqref{eq:upperboundong}  of Theorem \ref{pro:firstbound} holds. Thus
$(H_g -z, H_0 - z)$ is a Feshbach pair for $\boldsymbol{\chi}_{\rho}^{(0)}$ provided $g$ is in a neighborhood of zero and $z$ is
in a neighborhood of the unperturbed ground state energy. Thus by Theorem \ref{pro:firstbound}
 we can expand the resolvent in the Feshbach operator in a absolutely convergent Neumann series   \eqref{eq:firstresolventexp}. 
We shall bring the resulting Neumann series into normal order by using the pull-through formula,
see Lemma \ref{lem:pullthrough}, and applying a generalized version of Wicks theorem, see Theorem \ref{thm:wicktheorem}.
First we use an  alternative notation for our original Hamiltonian and  we define
\begin{align}
& w^{(0)}_{g,0,0}(z)(r) := H_{\rm at} - z + r ,  \nonumber \\
&w^{(0)}_{g,1,0}(z)(r, {k})   :=  g   G({k}  )  , \label{defofwI} \\
&w^{(0)}_{g,0,1}(z)(r, \tilde{{k}}) :=  g  G^*(\tilde{{k}}) . \nonumber
\end{align}
We set  $w_g^{(0)}:=(w^{(0)}_{g,m,n})_{0 \leq m+n \leq 1 }$ and use  the notation
\begin{align}
 \label{eq:defhlinemn}
& {H}_{m,n}^{^{(0)}}(w_{m,n}) := \int_{{A}^{m+n}} \frac{ dK^{(m,n)}}{|K^{(m,n)}|^{1/2}}
a^*(k^{(m)}) w_{m,n}(H_f, K^{(m,n)}) a(\widetilde{k}^{(n)}) .
\end{align}
We note that the expression  is analogous to   \eqref{eq:defofhmn}, apart from the fact that
the domain of integration is different and there is no projection. To distinguish this difference
 we use a   superscript  zeroth order.  In the new notation the interaction reads
\begin{equation} \label{eq:newwayint}
g W = H_{1,0}^{^{(0)}}(w_{g,1,0}^{(0)}) +  H_{0,1}^{^{(0)}}(w_{g,0,1}^{(0)}) .
\end{equation}
For the bookkeeping of the terms in the Neumann expansion we introduce the following multi-indices for $L \in \N$
\begin{align*}
	\underline{m} &:= (m_1,...,m_L) \in \N_0^{L }, \\
	|\underline{m}|  &:= m_1 + \cdots + m_L   , \\
	\underline{0} &:= (0,...,0) \in \N_0^{L } . 
\end{align*}
Now inserting the alternative  expression for the interaction, \eqref{eq:newwayint}, into  the convergent Neumann Series \eqref{eq:firstresolventexp},
and using the generalized Wick theorem, Theorem   \ref{thm:wicktheorem},   we
obtain a sum of terms of the form
\begin{align} \label{eq:defofv}
&V_{\umm,\upp,\unn,\uqq}^{(0,\rho)}[w](r, K^{(|\umm|,|\unn|)})\\
&\quad :=
(P_{\rm at}  \otimes P_\Omega )  F_0^{(0,\rho)}[w](H_f + \rho( r +\tilde{r}_0) )\nonumber \\
 &\qquad\quad\;\, \prod_{l=1}^L \left\{{{W}}_{\quad p_l,q_l}^{^{(0)}m_l,n_l}[w](\rho K^{(m_l,n_l)})
 F_l^{(0,\rho)}[w](H_f + \rho ( r + \tilde{r}_l) ) \right\} (P_{\rm at}  \otimes P_\Omega), \nonumber
\end{align} 
where  the definition of $\widetilde{r}_l$ is given in   \eqref{eq:rltildedef},
and where we used the definitions
\begin{align} \label{eq:defofWW1}
&{{W}}_{\quad p,q}^{^{(0)}m,n}[w](K^{(m,n)}) \\
	&\quad\; := \!\int \displaylimits_{(\R^3 \times  \{1,2\} )^{p+q}}
	\!\!a^*(x^{(p)}) \,
		w_{m+p,n+q}[k^{(m)}, x^{(p)} , \tilde{k}^{(n)} , \tilde{x}^{(q)} ] \,
		a(\tilde{x}^{(q)})\, \frac{d X^{(p,q)}}{|X^{(p,q)}|^{1/2}}, \nonumber
\end{align}
\begin{align*}
F_l^{(0,\rho)}[w](r) &:= F^{(0,\rho)}[w](r) :=  \frac{ {(\boldsymbol{\chib}_{\rho}^{(0)})}^2(r) }{w_{0,0}(r)  },
\qquad \textrm{ for } l=1,...,L-1 ,  \\
  F_0^{(0,\rho)}[w](r) & := F_L^{(0,\rho)}[w](r) := \chi(r/\rho )   .
\end{align*}
We use the natural  convention that there is no integration if  $p=q=0$  and that the
argument  $K^{(m,n)}$ is dropped if $m=n=0$.
Note that the appearance of the  $\rho$'s in the arguments on the r.h.s. of  \eqref{eq:defofv} is due to the
scaling transformation $S_\rho$ in eq. \eqref{eq:initialFeshbachOperator}.
Thus we have shown  the algebraic part  of  the following result.

\begin{proposition} \label{intkernelrep1st} Suppose the assumptions of  Theorem \ref{pro:firstbound}  hold, i.e.,  let $0 < \rho \leq \frac{1}{4}$,  $z \in D_{\rho/2}(\epsilon_{\rm at})$ and  \eqref{eq:upperboundong}.
Define
\begin{align}\label{eq:defofwmnschlange(11)}
\hat{w}^{(1,\rho)}_{g, 0,0}&(z)(r) \\ &:=  \rho^{-1} \left( \epsilon_{\rm at} - z + \rho r   + \sum_{L=2}^\infty (-1)^{L+1}
\sum_{\upp,\uqq \in \N_0^{L}: p_l+q_l = 1}
V_{\uzz,\upp,\uzz,\uqq}^{(0,\rho)}[w_g^{(0)}(z)](r) \right) \nonumber  ,
\end{align}
and for  $M,N \in \N_0$ with  $M+N \geq 1$  define
\begin{align} \label{eq:defofwmnschlange(0)}
 \hat{w}^{(1,\rho)}_{g, M,N}&(z)(r , K^{(M,N)}) \\
 &:= \sum_{L=1}^\infty (-1)^{L+1} \rho^{M+N - 1 }
\sum_{\substack{ \umm,\upp,\unn,\uqq  \in \N_0^{L}: \\ |\umm|=M, |\unn|=N, \\  m_l+p_l+q_l+n_l =1 } } 
V_{\umm,\upp,\unn,\uqq}^{(0,\rho)}[w_g^{(0)}(z)](r,K^{(M,N)}).  \nonumber
\end{align}
Assume that  the right hand sides converge with respect to the norm  $\| \cdot \|_{\mu,\xi}^\#$ for some $\mu  >  0$ and $\xi \in (0,1)$.  Then
for the symmetrization ${w}_g^{(1,\rho)}(z) := [\hat{w}_g^{(1,\rho)}(z)]^{\rm sym}$
we have
$$
{H}_g^{(1,\rho)}(z)  = H({w}_g^{(1,\rho)}(z)) .
$$
\end{proposition}
\begin{proof} Let $\hat{w}_g^{(1,\rho,L_0)}$ be defined as the integral kernel obtained  by the right hand sides of  \eqref{eq:defofwmnschlange(0)}  and \eqref{eq:defofwmnschlange(11)}, if we sum $L$ only up to $L_0$.
Then by the absolute convergence of the Neumann Series  \eqref{eq:firstresolventexp} and the application of the generalized Wick theorem, Theorem   \ref{thm:wicktheorem}, as discussed above, we find
$$
{H}_g^{(1,\rho)}(z) = \lim_{L_0 \to \infty} H(\hat{w}_g^{(1,\rho,L_0)}) .
$$
A detailed description how to obtain the integral kernels
is given in Appendix A of \cite{GriHas09}, see also \cite{BCFS}.
The assumption that the right hand sides of  \eqref{eq:defofwmnschlange(0)}  and \eqref{eq:defofwmnschlange(11)} converge with respect to the norm  $\| \cdot \|_{\mu,\xi}^\#$ imply in view of Theorem \ref{thm:injective}
that
\begin{align*}
 \lim_{L_0 \to \infty} H(\hat{w}_g^{(1,\rho,L_0)}) &=  \lim_{L_0 \to \infty} H([\hat{w}_g^{(1,\rho,L_0)}]^{\rm sym}) \\ &= H( \lim_{L_0 \to \infty}  [\hat{w}_g^{(1,\rho,L_0)}  ]^{\rm sym})   =   H({w}_g^{(1,\rho)}(z)) .
 \tag* \qedhere
 \end{align*}
\end{proof}

Our next goal is to show Inequalities \eqref{eq:feb2:1}, \eqref{eq:feb2:2}, and \eqref{eq:feb2:3},
below. These estimates will on the one hand imply that
right hand sides of   \eqref{eq:defofwmnschlange(0)}  and \eqref{eq:defofwmnschlange(11)} converge with respect to the norm  $\| \cdot \|_{\mu,\xi}^\#$,
and on the other hand  establish a  proof of Theorem \ref{thm:inimain1}.
To obtain the desired estimates  we will use  the  bounds, collected  in the  following   lemma.

\begin{proposition} \label{initial:thmE22}
For  all  $G \in L^2_\mu(\R^3\times \Z_2 ; \mathcal{L}(\HH_{\rm at}))$,  $\rho \in (0,1/4)$,  $z \in D_{\rho/2}(\epsilon_{\rm at})$,  $L \in \N$, and
 $\umm,\upp,\unn,\uqq \in \N_0^L$ we have
\begin{align}  \label{initial:thmmain:eq2}
 \rho^{|\umm|+|\unn| - 1 } \| V_{\umm,\upp,\unn,\uqq}^{(0,\rho)}[w_g^{(0)}] \|_\mu^\#  
&   \leq
 ( L + 2 )    \hat{C}_F^{L-1} (1 + \| \chi' \|_\infty)   |g|^L        \rho^{-L + \frac{1}{2}(|\upp|-p_1 +|\uqq|- q_L)}   \\
 & \qquad \qquad \qquad\times \rho^{(1+ \mu)(|\umm|+|\unn|) }
  \left\|\frac{G}{\omega} \right\|^{|\upp|+|\uqq|}   \| G \|_{\mu}^{|\umm |+ |\unn|} \nonumber
\end{align}
 and
\begin{align} \label{initial:thmmain:eq3}
\rho^{- 1 } \| V_{\underline{0},\upp,\underline{0},\uqq}^{(0,\rho)}[w_g^{(0)}] \|_{\infty} &   \leq
 \hat{C}_F^{L-1}    |g|^L    \rho^{-L + \frac{1}{2}(|\upp| +|\uqq|)}  \left\|\frac{G}{\omega} \right\|^{|\upp|+|\uqq|}   , \\
 \rho^{- 1 } \| \partial_r  V_{\underline{0},\upp,\underline{0},\uqq}^{(0,\rho)}[w_g^{(0)}] \|_{\infty} &   \leq
 ( L + 1  ) \hat{C}_F^{L-1} (1 + \| \chi' \|_\infty)   |g|^L    \rho^{-L + \frac{1}{2}(|\upp| +|\uqq|)}  \left\|\frac{G}{\omega} \right\|^{|\upp|+|\uqq|} \! ,  \label{initial:thmmain:eq3v2}
\end{align}
 where   $\underline{0} \in \N_0^L$.
\end{proposition}
\begin{remark}{\rm
We note that in contrast to Eq.  \eqref{eq:basicestseond} in Lemma  \ref{lem:approx:V} (see also  \cite{BCFS} Lemma 3.10)  we have an additional  factor  $\rho^{\frac{1}{2}(|\upp|-p_1 +|\uqq|- q_L)} $,  which  yields  an improved estimate. The proof of
of Theorem \ref{thm:inimain1},
which we present, will need  this improved estimate.}
\end{remark}
To show the above  proposition we will use the estimates from the following lemma. 
\begin{lemma} \label{ini:lemelemestimates} 
For  $\rho \geq 0$ let  $B_\rho =  H_f + \rho$.
\begin{itemize}
\item[(a)] Let   $\omega^{-1/2} G \in L^2(\R^3 \times \Z_2 ; \mathcal{L}(\HH_{\rm at}))$.
Then  for all  $m,n,p,q \in \N_0$, with  $m+n+p+q=1$,   all ${K}^{(m,n)}\in {B}_1^{m+n}$, and $\rho \geq 0$
\begin{align}
 \| B_{\rho}^{-1/2} {{W}}_{\quad p,q}^{^{(0)}m,n}&[w_g^{(0)}] 
(K^{(m,n)}) B_{\rho}^{-1/2} \| \nonumber\\
	&\leq  \left\|\frac{G}{\omega} \right\|^{p+q}  |g| \{ \| G({k}_1) \| \}^m \{  \| G(\tilde{{k}}_1) \| \}^n
		\rho^{\frac{1}{2}( p + q    ) - 1   } , \label{eq1:ini:lemelemestimates}
\end{align}
\begin{align}
 \| 1_{[0,1]}(H_f) {{W}}_{\quad p,q}^{^{(0)}m,n}&[w_g^{(0)}](K^{(m,n)}) B_{\rho}^{-1/2} \| \nonumber\\
	&\leq  \left\| \frac{G}{\omega} \right\|^{p+q}  |g| \{ \| G(k_1) \|\}^m \{  \| G(\tilde{{k}}_1) \| \}^n
		\rho^{\frac{1}{2}( q    - 1  )    } , \label{eq1:ini:lemelemestimates2}
\end{align}
\begin{align}
 \| B_{\rho}^{-1/2} {{W}}_{\quad p,q}^{^{(0)}m,n}&[w_g^{(0)} ]
(K^{(m,n)}) 1_{[0,1]}(H_f)  \| \nonumber\\
	&\leq   \left\| \frac{G}{\omega} \right\|^{p+q}   |g| \{  \| G({k}_1) \| \}^m \{   \| G(\tilde{{k}}_1) \| \}^n
		\rho^{\frac{1}{2}( p  - 1  )   }  \label{eq1:ini:lemelemestimates3} .
\end{align}
\item[(b)]  For  all $\rho \in (0, 1/4)$,   $z \in D_{\rho/2}(\epsilon_{\rm at})$ and  $r \in [0,\infty)$  we have
\begin{align}
 \| B_{\rho}^{1/2} F^{(0,\rho)}[w_g^{(0)}(z) ](H_f + \rho r   ) B_{\rho}^{1/2} \| &\leq 5 \leq \hat{C}_F \label{eq2:ini:lemelemestimates} , \\
 \| B_{\rho}^{1/2} \partial_r F^{(0,\rho)}[w_g^{(0)}(z)](H_f  + r\rho ) B_{\rho}^{1/2} \| 
 &\leq 20 + 10 \| \chi' \|_\infty 
\leq  \hat{C}_F (1  + \| \chi' \|_\infty)    \label{eq3:ini:lemelemestimates} .
\end{align}
\end{itemize}
\end{lemma}
\begin{proof} Eq.  \eqref{eq1:ini:lemelemestimates} follows from the proof of  Lemma  \ref{lem:estonint}.
Eqns. \eqref{eq1:ini:lemelemestimates2} and \eqref{eq1:ini:lemelemestimates3} follow also from the proof of  Lemma  \ref{lem:estonint}
and  \eqref{eq:projineq1}.
Estimate \eqref{eq2:ini:lemelemestimates}
follows from Lemma \ref{lem:firstfreeresolvbound}. To show Estimate \eqref{eq3:ini:lemelemestimates}
we calculate the derivative
\begin{align} \label{partialF}
& \partial_r F^{(0,\rho)}[w_g^{(0)}](H_f + r\rho)   \\
& \qquad \qquad =     \frac{ 2 \boldsymbol{\chib}_{\rho}^{(0)}(H_f + r\rho ) (\boldsymbol{\chib}_{1}^{(0)})'(\rho^{-1} H_f + r )    }{H_{\rm at} + H_f + r\rho  - z }
+  \frac{{(\boldsymbol{\chib}_{\rho}^{(0)})}^2\!(H_f + r\rho) \,  \rho    }{(H_{\rm at} + H_f + r\rho  - z )^2} . \nonumber
\end{align}
To estimate the terms on the right hand side  we use again  Lemma \ref{lem:firstfreeresolvbound}   together with
\begin{equation}
	\left\|\frac{B_{\rho}^{1/2}}{B_{\rho + r \rho }^{1/2}}\right\| 
		\leq 1 \,  . \tag*{\qedhere}
\end{equation}
\end{proof}

\noindent {\it Proof of Proposition \ref{initial:thmE22}.} Let $B_\rho =  H_f + \rho$.
We estimate $\| V_{\umm,\upp,\unn,\uqq}^{(0,\rho)}[w^{(0)}_g] \|_\mu$ using
\begin{equation} \label{eq:babyestimate1}
| \langle \varphi_{\rm at} \otimes \Omega , A_1 A_2 \cdots A_n \varphi_{\rm at} \otimes \Omega \rangle | \leq \| A_1 \| \| A_2 \| \cdots \| A_n \| ,
\end{equation}
where $\| \cdot \|$ denotes the operator norm,
and Inequalities \eqref{eq1:ini:lemelemestimates}--\eqref{eq3:ini:lemelemestimates}.
First we have for $ r \geq 0$
\begin{align}
&\|  V_{\umm,\upp,\unn,\uqq}^{(0,\rho)}[w_g^{(0)}](r, K^{(|\umm|,|\unn|)}) \|
\nonumber \\ &\leq
\bigg\| ( P_{\rm at}  \otimes P_\Omega )
F_0^{(0,\rho)}[w_g^{(0)}](H_f + \rho( r + \tilde{r}_0))
{{W}}_{\quad p_1,q_1}^{^{(0)}m_1,n_1}[w_g^{(0)}](\rho K^{(m_1,n_1)})
B_{\rho}^{-1/2} \nonumber \\
& \quad \times  B_{\rho}^{1/2}
 F_1^{(0,\rho)}[w_g^{(0)}](H_f + \rho( r + \tilde{r}_1) )   B_{\rho}^{1/2}   \nonumber\\
&\quad \times  \prod_{l=2}^{L-1}  \left\{
 B_{\rho}^{-1/2}   {{W}}_{\quad p_l,q_l}^{^{(0)}m_l,n_l}[w_g^{(0)}](\rho K^{(m_l,n_l)}) B_{\rho}^{-1/2}
 B_{\rho}^{1/2}  F_l^{(0,\rho)}[w_g^{(0)}](H_f + \rho( r + \tilde{r}_l))  B_{\rho}^{1/2}  \right\}\nonumber \\
&\quad \times  B_{\rho}^{-1/2}  {{W}}_{\quad p_L,q_L}^{^{(0)}m_L,n_L}[w_g^{(0)}](\rho K^{(m_l,n_l)})  F_L^{(0,\rho)}[w_g^{(0)}](H_f + \rho( r + \tilde{r}_L)) (P_{\rm at}  \otimes P_\Omega) \bigg\| \nonumber \\
&\leq \hat{C}_F^{L-1}
 \left\|\frac{G}{\omega} \right\|^{|\upp|+|\uqq|}  |g|^L  \left[   \prod_{l=1}^{L}
 [ \| G( \rho {k}_{m_l})\| ]^{m_l} [ \|G( \rho \tilde{{k}}_{n_l})\|]^{n_l} \right] \nonumber \\
  &\qquad\qquad\quad \times  \left[  \prod_{l=2}^{L-1} \rho^{\frac{1}{2}( p_l + q_l    ) - 1   }  \right] \rho^{\frac{1}{2}(q_1 + p_L) -1} \nonumber \\
  &\leq  \hat{C}_F^{L-1}   \left\|\frac{G}{\omega} \right\|^{|\upp|+|\uqq|}  |g|^L       \rho^{-L+1}\rho^{
  \frac{1}{2}( |\upp|-p_1 + |\uqq|-q_L ) }
   \prod_{l=1}^L[ \| G( \rho {k}_{m_l})\| ]^{m_l} [ \| G( \rho \tilde{{k}}_{n_l})\|]^{n_l} . \label{banachfirstbasicder1}
\end{align} 
To calculate the derivative we use Leibniz rule and we obtain similarly
 using again equation \eqref{eq:babyestimate1}
and Inequalities \eqref{eq1:ini:lemelemestimates}--\eqref{eq3:ini:lemelemestimates}. We find for $r \geq 0$
\begin{align} \label{banachfirstbasicder2}
\|\partial_r  V_{\umm,\upp,\unn,\uqq}^{(0,\rho)}&[w_g^{(0)}](r, K^{(|\umm|,|\unn|)})\|
\\
& \leq \hat{C}_F^{L-1}(1 + \| \chi'\|_\infty) (L+1)
\left\|\frac{G}{\omega} \right\|^{|\upp|+|\uqq|} |g|^L
\rho^{-L+1 + \frac{1}{2}(|\upp|-p_1 +|\uqq|-q_L)} \nonumber \\
  &\qquad\qquad \times \prod_{l=1}^L
 [ \| G( \rho {k}_{m_l}) \|]^{m_l} [  \| G( \rho \tilde{{k}}_{n_l})\| ]^{n_l}
 \nonumber   .
\end{align}
 Now   we can  estimate  inserting \eqref{banachfirstbasicder1} and
 \eqref{banachfirstbasicder2}, respectively,
 \begin{align*}
 \| & V_{\umm,\upp,\unn,\uqq}^{(0,\rho)}[w_g^{(0)}] \|_\mu  \\
  & =    \left( \int_{{B}_1^{|\umm|+|\unn|}}
  \|  V_{\umm,\upp,\unn,\uqq}^{(0,\rho)}[w_g^{(0)}]  (  K^{(|\umm|,|\unn|)})                        \|_{\infty}^2
  \frac{d {K}^{(\umm,\unn)}}{|{K}^{(\umm,\unn)}|^{3 + 2 \mu}} \right)^{1/2}   \\
   & \leq    \hat{C}_F^{L-1}   |g|^L  \left\|\frac{G}{\omega} \right\|^{|\upp|+|\uqq|} \rho^{-L +1  + \frac{1}{2}(|\upp|-p_1 +|\uqq|-q_L)}     \\
  &\qquad\quad  \times
  \left(  \int_{{B}_1^{|\umm|+|\unn|}}
   \prod_{l=1}^L \left\{
  \| G( \rho {k}_{m_l}) \|^{ 2 m_l}  \| G( \rho \tilde{{k}}_{n_l})\| ]^{ 2 n_l} \right\}
  \frac{d {K}^{(\umm,\unn)}}{|{K}^{(\umm,\unn)}|^{3 + 2 \mu}} \right)^{1/2}  \\
  & \leq
  \hat{C}_F^{L-1}   |g|^L    \left\|\frac{G}{\omega} \right\|^{|\upp|+|\uqq|} \rho^{-L+1  + \frac{1}{2}(|\upp|-p_1 +|\uqq|-q_L)}        \rho^{ \mu (|\umm| + |\unn|)}
    \| G \|_{\mu}^{|\umm |+ |\unn|} ,
\end{align*}
\begin{align*}
 \| \partial_r & V_{\umm,\upp,\unn,\uqq}^{(0,\rho)}[w_g^{(0)}] \|_\mu \\
  & =   \left( \int_{{B}_1^{|\umm|+|\unn|}}
  \|  \partial_r V_{\umm,\upp,\unn,\uqq}^{(0,\rho)}[w_g^{(0)}]  (  K^{(|\umm|,|\unn|)})                        \|_{\infty}^2
  \frac{d {K}^{(\umm,\unn)}}{|{K}^{(\umm,\unn)}|^{3 + 2 \mu}} \right)^{1/2} \\
   & \leq   ( L + 1 ) \hat{C}_F^{L-1} (1 + \| \chi'\|_\infty)   |g|^L  \left\|\frac{G}{\omega} \right\|^{|\upp|+|\uqq|} \rho^{-L +1  + \frac{1}{2}(|\upp|-p_1 +|\uqq|-q_L)}       \\
  &\qquad\quad  \times
   \left( \int_{{B}_1^{|\umm|+|\unn|}}
   \prod_{l=1}^L \left\{
  \| G( \rho {k}_{m_l}) \|^{2 m_l}  \| G( \rho \tilde{{k}}_{n_l})\| ]^{ 2 n_l} \right\}
  \frac{d {K}^{(\umm,\unn)}}{|{K}^{(\umm,\unn)}|^{3 + 2 \mu}}\right)^{1/2}  \\
  & \leq
   ( L + 1 ) \hat{C}_F^{L-1} (1 + \| \chi' \|_\infty)   |g|^L    \left\|\frac{G}{\omega} \right\|^{|\upp|+|\uqq|} \\ 
   & \qquad\quad \times \rho^{-L+1  + \frac{1}{2}(|\upp|-p_1 +|\uqq|-q_L)}       \rho^{ \mu (|\umm| + |\unn|)}
    \| G \|_{\mu}^{|\umm |+ |\unn|} .
\end{align*}
Adding above estimates yields  \eqref{initial:thmmain:eq2}.
Eqns. \eqref{initial:thmmain:eq3} and  \eqref{initial:thmmain:eq3v2} follow similarly  noting that
$|\umm|=0$ and $|\unn|=0$ can only occur if $L$ is even and  on the very left we have a
annihilation operator and on the very right a creation operator.
\qed

\vspace{0.5cm}

\noindent {\it Proof of Theorem  \ref{thm:inimain1}.}
It suffices to  establish inequalities \eqref{eq:feb2:1}--\eqref{eq:feb2:3}, below.
Let $S^L_{M,N}$ denote the set of tuples $(\umm,\upp,\unn,\uqq) \in \N_0^{4L}$ with
$|\umm|=M$, $|\unn|=N$, and
\begin{equation} \label{eq:condonsum}
m_l+p_l+q_l+n_l = 1 .
\end{equation}
Such tuples obviously  satisfy  $|\umm|+|\unn|+|\unn|+|\uqq|=L$.
Using this identity, we  now estimate the norm of \eqref{eq:defofwmnschlange(0)} using \eqref{initial:thmmain:eq2} and  \eqref{eq:diffnormonGest}.
This yields
\begin{align}
\| &( w^{(1,\rho)}_{g,M,N})_{M+N \geq 1}(z) \|_{\mu,\xi}^\#   \nonumber \\
&= \sum_{M+N \geq 1} {\xi}^{-(M+N)} \| {w}^{(1,\rho)}_{g,M,N}(z) \|_\mu^\#   \nonumber \\
&\leq \sum_{M+N\geq 1} \sum_{L=1}^\infty \sum_{(\umm,\upp,\unn,\uqq) \in S^L_{M,N}}
{\xi}^{-(M+N)} \rho^{M+N-1}  \| V_{\umm,\upp,\unn,\uqq}^{(0,\rho)}[w_g^{(0)}(z)] \|_\mu^\# \nonumber \\
&\leq \sum_{L=1}^\infty \sum_{M+N\geq 1} \sum_{(\umm,\upp,\unn,\uqq) \in S^L_{M,N} }
{\xi}^{-|\umm|-|\unn|} ( L + 2 ) \hat{C}_F^{L-1}(1+\| \chi'\|_\infty) \nonumber \\ 
&\qquad\qquad \times |g|^L   \rho^{-L/2 + (\frac{1}{2} +\mu)(|\umm|+|\unn|) - \frac{1}{2}(p_1+q_L) }
          \| G \|_{\mu}^L 
           \nonumber  \\
&\leq    \rho^{\mu } (1+\| \chi' \|_\infty )  \bigg[ \rho^{1/2} 6  {\xi}^{-1}|g|  \rho^{-1/2}    \| G \|_{\mu}  + 
 {64}   \left( \xi^{-1} |g| \rho^{-1/2} \hat{C}_F  \| G \|_{\mu}  \right)^2 \nonumber \\
& \qquad\qquad\qquad\qquad +    \rho^{-1/2} \sum_{L=3}^\infty  (L+2)  \left( 4  \xi^{-1}    |g| \rho^{-1/2} \hat{C}_F   \| G \|_{\mu}   \right)^L \bigg]  \; , \label{eq:feb2:1}
\end{align}
where  in the last inequality we estimated   the summands  $L=1$,  $L=2$ separately  and summed over the  terms with   $L \geq 3$,
as we now outline.
First we note that  \eqref{eq:condonsum}  implies that $S_{M,N}^L$ is empty unless   $M + N  \leq L$, and  that the number of elements  $(\umm,\upp,\unn,\uqq) \in \N_0^{4L}$ with
$\eqref{eq:condonsum} $ is bounded above by $4^{L}$.
Specifically   for $L=1$, we have only two  terms:  $(m_1,p_1,n_1,q_1)$ equal $(1,0,0,0)$ or  equal $(0,0,1,0)$.
For $L=2$ we use that  $1 \leq M + N $ implies   $p_1+q_L \leq 1$. For $L \geq 3$ we use $|\umm |+ |\unn| -  ( p_1+q_L ) \geq -1  $.
These considerations  establish  \eqref{eq:feb2:1}.
We now estimate the norm of \eqref{eq:defofwmnschlange(11)} using \eqref{initial:thmmain:eq3v2},
by means of a similar  but simpler estimate
\label{eq:defofwmnschlange(-1)}
\begin{align}
\sup_{r \in [0,1]} |  \partial_r  w^{(1,\rho)}_{g, 0,0}(z)(r) -  1|
&\leq \rho^{-1}  \sum_{L=2}^\infty \sum_{\substack{ \upp,\uqq  \in \N_0^{L}: \\ p_l+q_l= 1}}
 \|  \partial_r  V_{\uzz,\upp,\uzz,\uqq}^{(0,\rho)}[w_g^{(0)}(z)] \|_{\infty} \nonumber \\
&\leq (1+\| \chi' \|_\infty )  \sum_{L=2}^\infty  (L+1)  \left(2    \left\|\frac{G}{\omega} \right\|  |g| \rho^{-1/2} \hat{C}_F \right)^L \; .  \label{eq:feb2:2}
\end{align}
Analogously we have using \eqref{initial:thmmain:eq3}
\begin{align}
|  w^{(1,\rho)}_{g,0,0}(z)(0) + \rho^{-1} (z - \epsilon_{\rm at})  |
& \leq \rho^{-1} \sum_{L=2}^\infty \sum_{\substack{ \upp,\uqq \in \N_0^{L}:\\ p_l+q_l= 1}}
\| V_{\uzz,\upp,\uzz,\uqq}^{(0,\rho)}[w_g^{(0)}(z)] \|_{\infty} \nonumber
\\ \label{eq:feb2:3}
& \leq  \sum_{L=2}^\infty    \left(2   \left\|\frac{G}{\omega} \right\|  |g|  \rho^{-1/2} \hat{C}_F \right)^L \; .
\end{align} 
The series on the
 right hand sides in \eqref{eq:feb2:1}--\eqref{eq:feb2:3} converge if  \eqref{eq:condongrho1} holds.
Thus Theorem \ref{thm:inimain1} now follows in view of  equations  \eqref{eq:feb2:1}--\eqref{eq:feb2:3}.
\qed

\section{Second Feshbach Step}
\label{sec:secondFeshStep}

In this section we perform our second Feshbach step.
We want to note that the field energy cutoff of  the first Feshbach step will be henceforth denoted by $\rho_0$ while
the field energy cutoff of the second Feshbach step will be denoted by 
$\rho_1$.
First we  approximate $w_{g,0,0}^{(1,\rho_0)}(z)$, which is the content of
the following lemma.
We recall the  mapping $Z_{\rm at}$, which  was defined in \eqref{eq:defofzat}, and
 that its ground state energy, $\epsilon^{(2)}_{\rm at}$, is by assumption a simple eigenvalue.

 \begin{lemma}[Free Approx. to  1st  Feshbach]\label{norm:w:(0b)} There exists a constant $C$ such that the following holds. Let
 $G \in L^2_\mu(\R^3\times \Z_2 ; \mathcal{L}(\HH_{\rm at}))$,
 $0 < \rho_0 < 1/4$, and suppose
\begin{align} \label{eq:upperboundongv2}
	|g| < \frac{ \rho_0^{1/2}}{4 C_F    \|\omega^{-1}G \| } .
\end{align}
If   $z \in   D_{\rho_0 / 2}(\epsilon_{\rm at})$,  then  the defining series  of  $w^{(1,{\rho_{0}})}_{g, 0,0}(z)$, i.e.,  the r.h.s. of \eqref{eq:defofwmnschlange(11)},
converges absolutely and
\begin{align*}
 \sup_{0 \leq r \leq 1} \|\rho_0^{-1} (\epsilon_{\rm at} - z + \rho_0  r +   \chi^2(r)  g^2 Z_{\rm at} ) -   w^{(1,{\rho_{0}})}_{g, 0,0}(z)(r)   \|  
 \leq C (
\| G \|_\mu |g|^2 +  \| G\|_\mu^4  C_F^4 |g|^4 \rho_0^{-2} )   .
\end{align*}
\end{lemma}

\begin{remark}
{\rm We note that if we replaced  the $\chi^2(r)$ in front of  $Z_{\rm at}$ by one, then the proof given below  would yield a bound  only  of order $ |g|^2 \rho_0^{-1}  $.  }
\end{remark}

\begin{proof}
From \eqref{eq:defofwmnschlange(11)} we see that 
\begin{align*}
 {w}^{(1,{\rho_{0}})}_{g, 0,0}(z)(r)  = \rho_0^{-1} \left( 
 \epsilon_{\rm at} - z + \rho_0  r   + \sum_{L=2}^\infty (-1)^{L+1}
\sum_{\upp,\uqq \in \N_0^{L}: p_l+q_l = 1}\
V_{(\uzz,\upp,\uzz,\uqq)}^{(0,{\rho_{0}})}[w^{(0)}_g(z)](r) \right) \!.
\end{align*}
From the definition   \eqref{eq:defofv} we see, that  the summand  with $L=2$ is only non-vanishing if $\upp = (0,1)$ and $\uqq = (1,0)$ and that summands with odd $L$ vanish.
Using this   we can write
\begin{align}
  \rho_0^{-1} &(\epsilon_{\rm at} - z + \rho_0 r +   \chi^2(r)   g^2 Z_{\rm at} ) -   w_{g,0,0}^{(1,{\rho_{0}})}(z)(r)  \nonumber \\
 &=  \rho_0^{-1} \bigg( \chi^2(r)   g^2 Z_{\rm at}   -   X_g^{\rho_0}(z)(r) \bigg)  -    Y_g^{\rho_0}(z)(r) \nonumber  \\
& =  \rho_0^{-1} \bigg( \chi^2(r)   g^2 Z_{\rm at}   -   X_g^{\rho_0}(\epsilon_{\rm at})(r) +  X_g^{\rho_0}(\epsilon_{\rm at})(r)  -   X_g^{\rho_0}(z)(r) \bigg)      -    Y_g^{\rho_0}(z)(r) \label{norm:w:(0b)eq:1}
\end{align}
where we introduced the notation
\begin{align*}
  X_g^{\rho_0}(z)(r)  & :=   -   V_{(\uzz,(0,1),\uzz,(1,0))}^{(0,{\rho_{0}})}[w^{(0)}_g(z)](r)    \\
    Y_g^{\rho_0}(z)(r)  & :=   \rho_0^{-1}    \sum_{L=4}^\infty (-1)^{L+1}
\sum_{\upp,\uqq \in \N_0^{L}: p_l+q_l = 1}
V_{(\uzz,\upp,\uzz,\uqq)}^{(0,{\rho_{0}})}[w^{(0)}_g(z)](r)    .
\end{align*}
The second term  can be estimated similarly to  \eqref{eq:feb2:3}, i.e., using \eqref{initial:thmmain:eq3} we find
\begin{align}  \label{eq:zeroapprox0}
\sup_{r \in [0,1]} \|  Y_g^{\rho_0}(z)(r)   \|
\leq   \sum_{L=4}^\infty
\left(2  |g|  \left\|\frac{G}{\omega} \right\| \hat{C}_F \rho_0^{-1/2} \right)^L
\end{align}
In order to obtain a suitable estimate  for   $X_g^{\rho_0}$  we use that   $(\boldsymbol{\chib}_{\rho_0}^{(0)})^2$ has a natural decomposition into a sum of two terms
and we calculate the vacuum expectation  using the pull-through formula
\begin{align}
 X_g^{\rho_0}(z)(r) 
  = &\, g^2  \chi(r) P_{\rm at}  \int_{} \frac{ dk}{|k|}    G^*(k)    {P}_{\rm at}  \frac{ \chib^2( \rho_0^{-1} |k| + r  ) }{  \epsilon_{\rm at}  - z +  |k| + \rho_0  r  }    P_{\rm at}  G(k)P_{\rm at}  \chi(r) \nonumber  \\
	&\quad  +  g^2  \chi(r) P_{\rm at} 
 \int_{} \frac{ dk}{|k|}   G^*(k) \overline{P}_{\rm at} 
  \frac{1}{ H_{\rm at} - z +   |k|  +  \rho_0  r  } 
  \overline{P}_{\rm at} G(k) P_{\rm at}  \chi(r).   \label{eq:zeroapprox1} 
\end{align}
First we estimate the relative error if we replace $z$ by $\epsilon_{\rm at}$, that is
we show for $z \in D_{\rho_0/2}(\epsilon_{\rm at})$
 \begin{equation}  \label{eq:mainapprox1}
 \sup_{r\in[0,1]} \| R_g^{\rho_0}(z)(r) - R_g^{\rho_0}(\epsilon_{\rm at})(r) \| \leq    |g|^2 \rho_0 ( 3 + 1) \| G \|_\mu
 \end{equation}
 To estimate  the first term in   \eqref{eq:zeroapprox1} we use common denominators
 \begin{align}
 &  \frac{1}{  \epsilon_{\rm at}  - z + |k| +  \rho_0 r }    - \frac{1}{  |k| +  \rho_0 r } \label{eq:estfreeone}  \\
 &\qquad\qquad\qquad =     \frac{1}{  \epsilon_{\rm at}  - z + |k| + \rho_0  r}  (z - \epsilon_{\rm at}  )  \frac{1}{   |k| + \rho_0 r }  \nonumber  \\
  &\qquad\qquad\qquad =   (z - \epsilon_{\rm at}  )    \left[ 1 + \frac{ z - \epsilon_{\rm at} }{\epsilon_{\rm at} - z +  |k| + \rho_0  r    }    \right] \frac{1}{ (  |k| +   \rho_0  r     )^{2} }  . \nonumber
 \end{align}
 This yields
  \begin{align*}
\chib^2( \rho_0^{-1} |k| + r  ) \left|   \text{ l.h.s. of } \eqref{eq:estfreeone}  \right|   \leq    \rho_0 \left( 1 + \frac{\rho_0}{\frac{3}{4}\rho_0 -\frac{1}{4} \rho_0} \right) |k|^{-2} .
 \end{align*}
 This explains  the first contribution  on the right hand side of  \eqref{eq:mainapprox1}.
We estimate the second term in   \eqref{eq:zeroapprox1} similarly. For $E  \in \sigma(H_{\rm at} ) \setminus \{ \epsilon_{\rm at} \}$ we write
 \begin{align}
  &\frac{1}{ E  - z + |k| +  \rho_0 r  }  -
  \frac{1}{ E - \epsilon_{\rm at}  +    |k| +  \rho_0 r } \label{eq:estfreeotwo} \\
 &\quad  =
 \frac{1}{  E  - z + |k| +  \rho_0 r  }   (z - \epsilon_{\rm at}  )
 \frac{1}{  E - \epsilon_{\rm at} +   |k| +  \rho_0 r   }  \nonumber 
 \end{align}
  This yields
   \begin{align*}
 \left| \text{ l.h.s. of } \eqref{eq:estfreeotwo})  \right|   \leq    \rho_0  |k|^{-2} .
 \end{align*}
  This explains  the second  contribution  on the right hand side of  \eqref{eq:mainapprox1}.
 Next  we show that
 \begin{equation} \label{eq:mainapprox2}
 \sup_{r\in[0,1]}\| \chi^2(r) g^2 Z_{\rm at} - R_g^{\rho_0}(\epsilon_{\rm at} )(r) \| \leq  3  \| G \|_\mu |g|^2 \rho_0  .
 \end{equation}
 To  estimate  the first term in \eqref{eq:zeroapprox1} with $z = \epsilon_{\rm at}$  we use
  \begin{align*}
   \frac{1}{ |k| +  \rho_0 r  }    - \frac{1}{ |k| }   =
  \frac{1}{ |k| +  \rho_0 r }      ( - \rho_0 r   ) \frac{1}{   |k|  }  
   \end{align*}
and make use of
$$
|  \chib^2(\rho_0^{-1} |k| + r ) -  1   | \leq   \left\{ \begin{array}{ll} 0 ,  & |k| \geq \rho_0 \\ 1 ,  &  |k| \leq \rho_0  . \end{array} \right.
$$
 We estimate    the second term  in \eqref{eq:zeroapprox1} with $z = \epsilon_{\rm at}$   using for $E  \in \sigma(H_{\rm at} ) \setminus \{ \epsilon_{\rm at} \}$ that
   \begin{align*}
 \frac{1}{E - \epsilon_{\rm at}  +   |k| +  \rho_0 r  }    - \frac{1}{ E  - \epsilon_{\rm at}  +   |k|  }
   =
 \frac{1}{ E  - \epsilon_{\rm at}   +   |k| +  \rho_0 r  }   ( -  \rho_0 r   ) \frac{1}{ E  - \epsilon_{\rm at}   +   |k|   } .
 \end{align*}
 This gives   \eqref{eq:mainapprox2}. Finally,
inserting  estimates of  \eqref{eq:zeroapprox0},   \eqref{eq:mainapprox1}, and  \eqref{eq:mainapprox2} into  \eqref{norm:w:(0b)eq:1}  shows
the lemma.
\end{proof}

Let $P_{{\rm at}}^{(2)}$ denote the projection onto the one dimensional eigenspace of $Z_{\rm at}$ with eigenvalue $\epsilon_{\rm at}^{(2)}$ and let
$
\overline{P}_{\rm at}^{(2)} = 1-P_{\rm at}^{(2)}$. We mention that the superscript $(2)$ originates from the fact that these expressions are obtained
by  formal second order perturbation theory.
For  $\rho_1 >0$ define
\begin{align*}
\boldsymbol{\chi}^{(1)}_{\rho_1}(r)  &= P_{\rm at}^{(2)} \otimes \chi(r / \rho_1) \\
\boldsymbol{\chib}^{(1)}_{\rho_1}(r)  &= \overline{P}_{\rm at}^{(2)}  \otimes 1 + P_{\rm at}^{(2)}  \otimes \chib(r / \rho_1)
\end{align*}
and
\begin{align}
\boldsymbol{\chi}^{(1)}_{\rho_1}  &=      \boldsymbol{\chi}^{(1)}_{\rho_1}(H_f)       \label{eq:defofsecproj1}      \\
\boldsymbol{\chib}^{(1)}_{\rho_1}   &=  \boldsymbol{\chib}^{(1)}_{\rho_1}(H_f) .  \label{eq:defofsecproj2}
\end{align}
Recall that we assumed that the distance, $d_{\rm at}$,  from the lowest to the second lowest eigenvalue of $H_{\rm at}$ is
 one.  By the  assumption $0 \leq \delta_0 < \pi/2$ the following expression is positive
 \begin{equation} \label{eq:secondfeshestm0}
c_{\delta_0}  := \inf_{g \in S_{\delta_{0}}  } | d_{\rm at} + g^{-2} |   > 0  ,
\end{equation}
which follows from an easy minimization problem, yielding
$$
c_{\delta_0} = \left\{ \begin{array}{ll} d_{\rm at} ,  & \text{ if }  0 \leq \delta_0  \leq \pi/4 \\ d_{\rm at} \sin ( \pi - 2 \delta_0 )  ,  &  \text{ if } \pi/4 < \delta_0   <  \pi/2  .  \end{array} \right.
$$
Let $\rho_0 > 0$. We shall assume the following inequalities
 \begin{equation} \label{eq:secondfeshestm1}  
 \rho_0^{-1} |g|^2 <  \frac{\frac{1}{4}}{ \| Z_{\rm at}  \|  +  c_{\delta_0}  }
\end{equation}
and
\begin{equation} 
 \label{eq:secondfeshest2}
\rho_1 \rho_0 \leq  |g|^2 c_{\delta_0}
\end{equation}
hold. 
Next  we show the following lemma, which will establish the required invertibility.

\begin{lemma}[Invertibility of Free Approx to 1st  Feshbach]\label{T-invert}  Suppose $\rho_0 , \rho_1 \in (0,1/2]$. 
Let   $g  \in S_{\delta_0}$ satisfy \eqref{eq:secondfeshestm1} and \eqref{eq:secondfeshest2}.
Then for  $z \in   D_{\rho_0 \rho_1/2}(\epsilon_{\rm at}  + g^2 \epsilon_{\rm at}^{(2)})$
we have
\begin{equation} \label{eq:Tg(a)invertible}
\left\| \left( \rho_0^{-1}   (\epsilon_{\rm at}  - z + \rho_0  H_f  +   \chi^2(H_f)  g^2 Z_{\rm at}  ) \upharpoonright  \ran \boldsymbol{\chib}^{(1)}_{\rho_1} \right)^{-1} \right\| \leq \frac{4}{\rho_1}
 .
\end{equation}
\end{lemma}

\begin{proof}
For notational simplicity we shall write
$$
X(\rho_0,g,z) :=   \rho_0^{-1} (\epsilon_{\rm at} - z + \rho_0  H_f  +   \chi^2(H_f)  g^2 Z_{\rm at}  ) .
$$
For normalized $\psi$ in the range of  $Q_1 := \overline{P}_{\rm at}^{(2)} \otimes  1_{[0,1]}(H_f)$
we have
\begin{align}
\| X(\rho_0,g,z)   \psi \| 
&\geq  \inf\limits_{\eta \in \ran \overline{P}_{\rm at}^{(2)} , \,  \| \eta \| = 1}
 \inf\limits_{ 0\leq r \leq 1 }
\| ( - \rho_0^{-1} g^2 \epsilon_{\rm at}^{(2)}   +    r  +   \rho_0^{-1} \chi^2(r) g^2 Z_{\rm at}  ) \eta  \|  \nonumber \\
&\qquad\qquad - |  \rho_0^{-1}  (
\epsilon_{\rm at}  + g^2 \epsilon_{\rm at}^{(2)}  - z ) |  \nonumber \\
& \geq \rho_0^{-1} |g|^2 c_{\delta_0}   - \rho_1/2 \,, \label{secstepfeshbasic}
\end{align}
where in the last inequality  we used on the one hand that for $r \in [0,3/4]$ we have by  \eqref{eq:secondfeshestm0}
\begin{align*}
   \| ( - \rho_0^{-1} g^2 \epsilon_{\rm at}^{(2)}   +    r +   \rho_0^{-1} \chi^2(r) g^2 Z_{\rm at} ) \eta  \|
    &=  \rho_0^{-1} | g |^2  \| (  - \epsilon_{\rm at}^{(2)}   +  g^{-2}   \rho_0  r +   Z_{\rm at}  ) \eta  \|  \\
    &\geq \rho_0^{-1} |g|^2 c_{\delta_0}  ,
\end{align*}
and that on the other hand   we used that for $r \in [3/4,1]$ we have
$$
   \| ( - \rho_0^{-1} g^2 \epsilon_{\rm at}^{(2)}   +     r  +   \rho_0^{-1} \chi^2(r) g^2 Z_{\rm at}  ) \eta  \| \geq
3/4 -  \rho_0^{-1}|g|^2    \| Z_{\rm at} -  \epsilon_{\rm at}^{(2)}  \|  \geq   \rho_0^{-1} |g|^2 c_{\delta_0}
$$
by
\eqref{eq:secondfeshestm1}.
Using  \eqref{eq:secondfeshest2}  in  \eqref{secstepfeshbasic}  it now follows that
$$
\| (X(\rho_0,g,z)  \upharpoonright \ran Q_1  )^{-1}  \| \leq \frac{2}{\rho_1} .
$$
Next we consider for normalized  $\psi$ in the range of  $Q_2 := {P}_{\rm at}^{(2)}  \otimes  \chib(H_f/\rho_1) $
\begin{align*}
 \|  X(\rho_0,g,z)   \psi \| 
& \geq  \inf_{\substack{ \eta \in \ran P_{\rm at}^{(2)}, \| \eta \| = 1 , \\   \, 3 \rho_1 / 4   \leq  r \leq 1}}
  \| (-  \rho_0^{-1} g^2 \epsilon_{\rm at}^{(2)}   +    r +  \chi^2(r)  \rho_0^{-1} g^2 \epsilon_{\rm at}^{(2)}    ) \eta  \|  
  \nonumber \\
 &\qquad\qquad - |  \rho_0^{-1}  (
\epsilon_{\rm at}  + g^2 \epsilon_{\rm at}^{(2)}  - z ) | \\
& \geq  \frac{3}{4} \rho_1 - \frac{1}{2} \rho_1 = \frac{1}{4} \rho_1 ,
\end{align*}
 where
we used that on the one hand side we have  for $r \in  [3 \rho_1 / 4 , 3/4]  $
$$
  \| ( - \rho_0^{-1} g^2 \epsilon_{\rm at}^{(2)}  +   r +  \chi^2(r)  \rho_0^{-1} g^2 \epsilon_{\rm at}^{(2)}    ) \eta  \| = r \geq
 \frac{3}{4} \rho_1 ,
$$
and on the other hand we have  for $r \in  [3  / 4 , 1]  $
$$
  \| ( - \rho_0^{-1} g^2 \epsilon_{\rm at}^{(2)}   +    r +  \chi^2(r)  \rho_0^{-1} g^2 \epsilon_{\rm at}^{(2)}    ) \eta  \|
\geq   \frac{3}{4} - \rho_0^{-1} |g|^2  | \epsilon_{\rm at}^{(2)} |  \geq  \frac{1}{2} \geq \frac{3}{4} \rho_1 ,
$$
by \eqref{eq:secondfeshestm1}.
Thus we can invert the operator $X(\rho_0,g,z) $ on the range of $\boldsymbol{\chib}^{(1)}_{\rho_1}$.
\end{proof}

In order to perform a second Feshbach iteration, we choose the following
 decomposition
\begin{equation}  \label{eq:decompsecond}
{H}_g^{(1,\rho_0)} = T_{g}^{(1,\rho_0)}    +   W_{g}^{(1,\rho_0)}
\end{equation}
where
\begin{align}
 T_{g}^{(1,\rho_0)}   & := H_{0,0}( t_{g}^{(1,\rho_0)}) , \\
 W_{g}^{(1,\rho_0)} & :=  H(w_{g,{\rm int}}^{(1,\rho_0)})    \label{eq:defofwga}
\end{align}
where
\begin{align*}
 t_{g}^{(1,\rho_0)} &  :=    {P}_{\rm at}^{(2)}  w_{g,0,0}^{(1,\rho_0)} {P}_{\rm at}^{(2)} +  \overline{P}_{\rm at}^{(2)}  w_{g,0,0}^{(1,\rho_0)} \overline{P}_{\rm at}^{(2)}   \\
 w_{g,{\rm int}}^{(1,\rho_0)}  &:= ({P}_{\rm at}^{(2)}  w_{g,0,0}^{(1,\rho_0)} \overline{P}_{\rm at}^{(2)} +  \overline{P}_{\rm at}^{(2)}  w_{g,0,0}^{(1,\rho_0)} {P}_{\rm at}^{(2)} ,   w_{g,m+n \geq 1}^{(1,\rho_0)} )
\end{align*}

\begin{remark} {\rm We note that decomposition  \eqref{eq:decompsecond} into
free    part  and interacting part is not unique. The isospectrality property of the  smooth Feshbach merely requires that the free part commutes
 with the smoothed projections. This  issue is pointed out  in Remark 2.4 in \cite{BCFS}.
Alternatively, we  could use the decomposition
according to
\begin{align}
t_{g,{\rm free} }^{(1,\rho_0)}(z)(r) &  :=  \rho_0^{-1} (\epsilon_{\rm at} - z + \rho_0 r +   \chi^2(r)  g^2 Z_{\rm at} ) \label{eq:tfreeapprox} \\ 
w_{g, {\rm rest }}^{(1,\rho_0)} &:=    ( w_{g,0,0}^{(1,\rho_0)}  - t_{g,{\rm free}}^{(1,\rho_0)}      , w_{g,m+n \geq 1}^{(1,\rho_0)} ) . \nonumber
 \end{align}
The  proof given in this paper would carry through also with this decomposition, with only notational modifications.
 }
\end{remark}

Thus we  can now prove the main result of this section.

\begin{theorem}[Abstract  Feshbach Pair Criterion  for 2nd  Iteration] \label{prop:feshsecond}
Assume that  the smallest eigenvalue of  $Z_{\rm at}$ is  simple.   Let  $\rho_1 \in (0,1]$.
Suppose $$t   \in C([0,1];\mathcal{L}({P}_{\rm at}^{(2)} \HH_{\rm at}) \oplus \mathcal{L}(\overline{P}_{\rm at}^{(2)}{P}_{\rm at} \HH_{\rm at}))$$
and   $w \in \WW_\xi^{[d]}$.
Then  the operators
$H_{0,0}(t)$ and  $H(w)$  are  a Feshbach pair for $\boldsymbol{\chi}^{(1)}_{\rho_1}$,
provided
\begin{itemize}
\item[(i)] $ H_{0,0}(t)$   is invertible on  the closure of   $\ran \boldsymbol{\chib}^{(1)}_{\rho_1} $  and    
$$  \| ( H_{0,0}(t) \upharpoonright \ran \boldsymbol{\chib}^{(1)}_{\rho_1} )^{-1}  \| \leq \frac{8}{\rho_1}\,,
$$
\item[(ii)] the following inequality holds
\begin{equation} \label{eq:condonsecfesh}  \| H(w)   \|  <  \frac{\rho_1}{8} .
\end{equation}
\end{itemize}
In this case we have the absolutely convergent expansion
\begin{align} \label{eq:neumannsecond}
 & F_{\boldsymbol{\chi}^{(1)}_{\rho_1}}( H_{0,0}(t)    , H(w)  ) \\
 & \qquad \qquad =  H_{0,0}(t)   +   \sum_{L=1}^\infty (-1)^{L-1} \boldsymbol{\chi}^{(1)}_{\rho_1} H(w)  \left( \frac{ (\boldsymbol{\chib}^{(1)}_{\rho_1})^2 }{
H_{0,0}(t) }    H(w)   \right)^{L-1} \boldsymbol{\chi}^{(1)}_{\rho_1}  . \nonumber
\end{align}
\end{theorem}
\begin{proof} First observe that $H_{0,0}(t)$ commutes with $\boldsymbol{\chi}^{(1)}_{\rho_1}$. The  Feshbach pair property follows
from (i) and (ii) and   Neumanns theorem.
The second claim follows again by Neumanns theorem.
\end{proof}

\begin{remark} {\rm
In the proof of the main theorem, in  Section \ref{sec:prooofmain},  we will determine an explicit   relation among  $\rho_0$ and $\rho_1$ and $g$.
Using this relation we will  verify assumption (i) and (ii)  of the above theorem by the help of   Lemmas \ref{norm:w:(0b)} and \ref{T-invert} .  }
\end{remark}

\section{Banach Space Estimate for the second Step}
\label{sec:secondStep}

Using estimates of the previous section we will show in Section~\ref{sec:prooofmain}
that  $(T_{g}^{(1,\rho_0)}(z), W_{g}^{(1,\rho_0)}(z))$
is indeed a Feshbach pair for   $\boldsymbol{\chi}^{(1)}_{\rho_1}$.
For the moment  we will assume that  the Feshbach property is satisfied.
In that case  we can define
\begin{equation*}
H_g^{(2,\rho_1)}(z) :=
	S_{\rho_1} ( F_{\boldsymbol{\chi}^{(1)}_{\rho_1}}
		( T_{g}^{(1,\rho_0)}(z)  ,  W_{g}^{(1,\rho_0)}(z)))
\end{equation*}
provided the right sides exist.
Similar to Section \ref{sec:firstStep} we now
want to show that there exists a sequence of integrals kernels 
$w_g^{(2,{\rho_{1}})}(z)$  such that
$$
      H( w^{(2,{\rho_{1}})}_g(z))
		=  {H}_g^{(2,{\rho_{1}})}(z)\restriction{\ran \boldsymbol{\chi}^{(1)}_{\rho_1}} .
$$
This will follow as a conclusion of the following theorem.

For notational compactness  we introduce in this section the following constant  $$C_\chi :=  20 \sqrt{2} .$$

\begin{theorem}[Abstract Banach Space Estimate for 2nd Feshbach Operator]\label{thm:ballprofforsecond}  Let  $0 < \xi \leq 1/4$ and
  assume that  the smallest eigenvalue of  $Z_{\rm at}$ is  simple.   Let  $\rho_1 \in (0,1]$.
Suppose $$t   \in C^1([0,1];\mathcal{L}({P}_{\rm at}^{(2)}   \HH_{\rm at}) \oplus \mathcal{L}(\overline{P}_{\rm at}^{(2)}{P}_{\rm at} \HH_{\rm at}))$$
and   $w \in \WW_\xi^{[d]}$.
\begin{itemize}
\item[(i)] $ H_{0,0}(t)$   is invertible on  the closure of   $\ran \boldsymbol{\chib}^{(1)}_{\rho_1} $  and  \\  
$  \| ( H_{0,0}(t) \upharpoonright \ran \boldsymbol{\chib}^{(1)}_{\rho_1} )^{-1}  \| \leq \frac{8}{\rho_1}$,
\item[(ii)]  $\| H(w)   \|  <  \frac{\rho_1}{8}    $.
\end{itemize} 
 Then $H(t)$ and $H(w)$ are a Feshbach pair for $ \boldsymbol{\chib}^{(1)}_{\rho_1} $.
Moreover, suppose
 $$
 \gamma < \frac{ \rho_1 }{8 C_\chi }
 $$
 and
\begin{align*}
& \| w \|_{\mu,\xi}^\#  \leq \gamma  , \\
&  \|t' \|_\infty \leq \tau_0  ,  \\
& \|{P}_{\rm at}^{(2)} t' {P}_{\rm at}^{(2)}  -1  \| \leq \tau_1 .
\end{align*}
Then
$$
 S_{\rho_1} ( F_{\boldsymbol{\chi}^{(1)}_{\rho_1}}  ( H(t), H(w) ))
 -   \rho_1^{-1} {P}_{\rm at}^{(2)} t(0)  {P}_{\rm at}^{(2)}   \in  \mathcal{B}^{[1]}(\alpha_1, \beta_1, \gamma_1 ) , $$
where
\begin{align*}
\alpha_1 & =  12 ( 1 + 2 \|\chi'\|_\infty + 8 \tau_0 )  C_\chi \gamma \rho_1^{-1} , \\ 
 \beta_1 & =    \tau_1  + 12 ( 1 + 2 \|\chi'\|_\infty + 8 \tau_0 ) C_\chi \gamma \rho_1^{-1} ,  \\ 
  \gamma_1 & =  96 (1 + 2 \|\chi'\|_\infty + 8 \tau_0  ) \rho_1^\mu C_\chi \gamma  . 
\end{align*}
\end{theorem}

In order to prove Theorem \ref{thm:ballprofforsecond} we first   use the Neumann expansion given in \eqref{eq:neumannsecond}  and we
normal order the resulting  expression using the generalized Wick Theorem, stated in the appendix.
The result is given in Proposition  \ref{prop:secondkernels} below, and can be obtained as in  \cite[Theorem 3.7]{BCFS}.
To state it  we recall the following notation
\begin{align*}
	\underline{m} &:= (m_1,...,m_L) \in \N_0^{L }, \\
	|\underline{m}|  &:= m_1 + \cdots + m_L   , \\
	\underline{0} &:= (0,...,0) \in \N_0^{L } ,
\end{align*}
for $L \in \N$.
We define
\begin{align} \label{eq:defofV}
&V_{\umm,\upp,\unn,\uqq}^{(1,\rho_1)}[t,w](r, K^{(|\umm|,|\unn|)}) \nonumber \\
&\qquad  := (P_{\rm at}^{(2)}  \otimes P_\Omega  )  F_0^{(1,\rho_1)}[t](H_f + \rho_1( r +\tilde{r}_0) )
\\
&\qquad\qquad\quad \prod_{l=1}^L \left\{
{W}_{p_l,q_l}^{m_l,n_l}[w](\rho_1 ( r + \tilde{r}_l), \rho_1 K^{(m_l,n_l)})
 F_l^{(1,\rho_1)}[t](H_f + \rho_1 ( r + \tilde{r}_l) ) \right\}  \nonumber\\
 &\qquad\qquad\qquad\qquad( P_{\rm at}^{(2)} \otimes P_\Omega  )  \,. \nonumber
\end{align}
For $w \in \WW_{m+p,n+q}^{[d]}$  we defined
\begin{align} \label{eq:defofWW2}
 &{W}_{p,q}^{m,n}[w](r,K^{(m,n)}) \nonumber\\
  &\qquad := 1_{[0,1]}(H_f)  \int_{{{B}_1}^{p+q}} \frac{d X^{(p,q)}}{|X^{(p,q)}|^{1/2}}  \,a^*(x^{(p)})  \\ 
  &\qquad \qquad\qquad
w_{m+p,n+q}[H_f + r, k^{(m)}, x^{(p)} , \tilde{k}^{(n)} , \tilde{x}^{(q)} ] \,a(\tilde{x}^{(q)})\, 1_{[0,1]}(H_f)\,,  \nonumber
\end{align}
where we use the natural  convention that there is no integration if  $p=q=0$  and that the argument $K^{(m,n)}$ is
dropped if $m=n=0$. Moreover we defined
 for $l=0,L$ the expressions  $F_l^{(1,\rho_1)}[t](r ) := \chi(r/\rho_1)$ and for $l=1,...,L-1$ we have set
\begin{equation*}
F_l^{(1,\rho_1)}[t](r) := F^{(1,\rho_1)}[t](r) =  \frac{ \left(\boldsymbol{\chib}_{\rho_1}^{(1)}(r)\right)^2 }{t(r)} ,
\end{equation*}
and  $\tilde{r}_l$ is defined as in \eqref{eq:rltildedef}.
We note that we used similar notation as in the previous section.

\begin{proposition} \label{prop:secondkernels}
Define
\begin{align}\label{w:1:g:00}
\hat{w}^{(2,{\rho_{1}})}_{ 0,0}(r) := \rho_1^{-1}
 \bigg(      t(\rho_1 r) + \sum_{L=1}^\infty (-1)^{L-1}
\sum_{\substack{ {{\upp,\uqq}} \in \N_0^{L} } }  
V_{{\underline{0},\upp,\underline{0},\uqq}}^{(1,\rho_1)}[t,w](r)\bigg) ,
\end{align}
and for    $M + N \geq 1$ define
\begin{align} \label{eq:defofwmnschlange(1)}
&\hat{w}^{(2,{\rho_{1}})}_{M,N}(r , K^{(M,N)}) \nonumber \\
&\quad := \sum_{L=1}^\infty (-1)^{L+1} \rho_1^{M+N - 1 }\!\!\!\!\!\!
\sum_{\substack{ {\umm,\upp,\unn,\uqq} \in \N_0^{L} \\ |\underline{m}|=M, |\underline{n}| = N  } }
\prod_{l=1}^L \left\{\binom{ m_l + p_l}{ p_l} \binom{ n_l + q_l}{ q_l } \right\}  V_{{\umm,\upp,\unn,\uqq}}^{(1,\rho_1)}[t ,w](r,K^{(M,N)}).
\end{align}
Assume that the right hand side converges with respect to the norm $\| \cdot \|_{\mu,\xi}^\#$.
Let  $ w^{(2,{\rho_{1}})}$ be the symmetrization w.r.t. $k^{(M)}$ and $\tilde{k}^{(N)}$ of $ \hat{w}^{(2,{\rho_{1}})}$.
 Then
$$
	S_{\rho_1} ( F_{\boldsymbol{\chi}^{(1)}_{\rho_1}} (H(t), H(w)))
		= H( w^{(2,{\rho_{1}})}) .
$$
\end{proposition}

The proof of this proposition is analogous to the proof of Proposition \ref{intkernelrep1st}.
To prove  Theorem \ref{thm:ballprofforsecond}, we shall need the  estimate given in the following lemma.

\begin{lemma} \label{lem:approx:V} Suppose the assumptions of Theorem \ref{thm:ballprofforsecond} hold.
For fixed $L \in \N$ and $\umm,\upp,\unn,\uqq \in \N_0^{L}$ and $w \in \WW_\xi^{[d]}$ 
we have
\begin{align} \label{eq:basicestseond}
\rho_1^{(|\umm|+|\unn|) - 1 } \| V_{{\umm,\upp,\unn,\uqq}}^{(1,\rho_1)}[t, w]  \|_\mu^\#   
& \leq (L+2) 2^{L/2} \hat{C}_\chi^{L-1}(1 + \| \chi'\|_\infty + 4 \|t' \|_\infty)  \\
& \qquad\qquad \times \rho_1^{(1 + \mu)(|\umm|+|\unn|)-L} \prod_{l=1}^L
		\!\frac{\|w_{m_l+p_l,n_l+q_l}\|_\mu^\# }{\sqrt{ p_l^{p_l}q_l^{q_l}}}.\nonumber
\end{align}
with  $\hat{C}_{\chi} = 20$ and the convention that $p^p := 1$  for $p = 0$.
\end{lemma}
\begin{remark}{\rm
We note that in contrast to \cite{BCFS} we do not have in \eqref{eq:defofwmnschlange(1)}
and \eqref{w:1:g:00} the conditions  $m_l+p_l+q_l+n_l \geq 1$ and $p_l+q_l \geq 1$, respectively.
The proof of Lemma \ref{lem:approx:V} is still similar to the proof of Lemma 3.10 in
\cite{BCFS} 
however we have to take into account more terms. These terms are hidden in our notation,
see Chapter \ref{sec:Banachspace}.} 
\end{remark}
\begin{proof}[Proof of Lemma \ref{lem:approx:V}]
We start by estimating the resolvents. Let $0\leq u + \rho_1 r \leq 1$ for $u, r \geq 0$.
Then for $l=0$ and $l=L$ we have
$$
|F_l^{(1,\rho_1)}[t](u+\rho_1 r)| \leq 1,
\qquad 
|\partial_r F_l^{(1,\rho_1)}[t](u+\rho_1 r)|
	\leq \|\chi'\|_\infty
$$
and for $l=1,...,L-1$,
\begin{align*}
\|F_l^{(1,\rho_1)}[t](u + \rho_1 r)\|  \leq \left\|\frac{ \left(\boldsymbol{\chib}_{\rho_1}^{(1)}(u + \rho_1 r)\right)^2 }{t(u + \rho_1 r)}\right\|
 \leq \frac{8}{\rho_1}
	 \leq \frac{\widehat{C}_\chi}{\rho_1}  ,
\end{align*}
\begin{align*}
 & \|\partial_r F_l^{(1,\rho_1)}[t](u + \rho_1 r)\|   \\
& \qquad  \leq
	 \left\|\frac{ 2 \boldsymbol{\chib}_{\rho_1}^{(1)}(u + \rho_1 r) \partial_r \boldsymbol{\chib}_{\rho_1}^{(1)}(u + \rho_1 r)}{t(u + \rho_1 r)}\right\|
 + \left\|\frac{ \left(\boldsymbol{\chib}_{\rho_1}^{(1)}(u + \rho_1 r)\right)^2 \rho_1 t'(u + \rho_1 r)}{(t(u + \rho_1 r))^2}\right\|
	\\
	& \qquad \leq  \frac{16}{\rho_1} \|\chi'\|_\infty + \frac{64}{\rho_1} \|t'\|_\infty  \leq \frac{\widehat{C}_\chi}{\rho_1}( \|\chi'\|_\infty + 4 \|t'\|_\infty  ) ,
\end{align*}
where we used a similar equation as \eqref{partialF}. \\
Now we estimate $\| V_{{\umm,\upp,\unn,\uqq}}^{(1,\rho_1)}[t, w] \|$ and
$\| \partial_r V_{{\umm,\upp,\unn,\uqq}}^{(1,\rho_1)}[t, w] \|$
using a variant of \eqref{eq:babyestimate1}.
\begin{align}
&	\|V_{{\umm,\upp,\unn,\uqq}}^{(1,\rho_1)}[t, w] (r, K^{(|\umm|,|\unn|)}) \|\nonumber \\
	&\qquad\leq \prod_{l = 0}^L \|F_l^{(1,\rho_1)}[t](H_f + \rho_1 ( r + \tilde{r}_l) )\|
		\prod_{l=1}^L \|{W}_{p_l,q_l}^{m_l,n_l}[w](\rho_1 ( r + \tilde{r}_l), \rho_1 K^{(m_l,n_l)})\|	\nonumber \\
 &\qquad\leq \widehat{C}_\chi^{L-1} \rho_1^{-L+1} \prod_{l=1}^L \left\|{W}_{p_l,q_l}^{m_l,n_l}[w](\rho_1 ( r + \tilde{r}_l), \rho_1 K^{(m_l,n_l)})\right\|. \label{eq:estonVinf}
	\end{align}
Similarly we get with Leibniz' rule
\begin{align}
	\|&\partial_r  V_{{\umm,\upp,\unn,\uqq}}^{(1,\rho_1)}[t, w]  (r, K^{(|\umm|,|\unn|)}) \|  \nonumber \\
	&\leq \left\{\sum_{j=0}^L \|\partial_r F_j^{(1,\rho_1)}[t](H_f + \rho_1 ( r + \tilde{r}_l) )\|
		\prod_{\substack{l = 0 \\ l \neq j}}^L \|F_l^{(1,\rho_1)}[t](H_f + \rho_1 ( r + \tilde{r}_l) )\| \right\}  \nonumber \\
			& \qquad \qquad \times \prod_{l=1}^L \|{W}_{p_l,q_l}^{m_l,n_l}[w](\rho_1 ( r + \tilde{r}_l), \rho_1 K^{(m_l,n_l)})\|
	 \nonumber	\\ &\qquad + \prod_{l = 0}^L \|F_l^{(1,\rho_1)}[t](H_f + \rho_1 ( r + \tilde{r}_l) )\| \nonumber \\
		& \qquad\qquad \qquad\quad \times \left\{\sum_{j=1}^L \|\partial_r {W}_{p_l,q_l}^{m_l,n_l}[w](\rho_1 ( r + \tilde{r}_l), \rho_1 K^{(m_l,n_l)})\| \right.  \nonumber \\
			 & \quad\qquad\qquad\qquad \qquad \qquad \qquad \left. \times \prod_{\substack{l = 1 \\ l \neq j}}^L \|{W}_{p_l,q_l}^{m_l,n_l}[w](\rho_1 ( r + \tilde{r}_l), \rho_1 K^{(m_l,n_l)})\|\right\}	\nonumber    \\
&\leq (L+1) \widehat{C}_\chi^{L-1} ( \| \chi' \|_\infty + 4 \| t' \|_\infty)  \rho_1^{-L+1} \nonumber \\
&\qquad  \prod_{l=1}^L \|{W}_{p_l,q_l}^{m_l,n_l}[w](\rho_1 ( r + \tilde{r}_l), \rho_1 K^{(m_l,n_l)})\|  \nonumber \\
		&\qquad\quad + \widehat{C}_\chi^{L-1} \rho_1^{-L+1} \left\{\rho_1 \sum_{j=1}^L {\|{W}_{p_l,q_l}^{m_l,n_l}[\partial_r w](\rho_1 ( r + \tilde{r}_l), \rho_1 K^{(m_l,n_l)})\|} \right. \nonumber   \\
		 & \quad\qquad \qquad\qquad \qquad \qquad \qquad \left. \times \prod_{\substack{l = 1 \\ l \neq j}}^L \|{W}_{p_l,q_l}^{m_l,n_l}[w](\rho_1 ( r + \tilde{r}_l), \rho_1 K^{(m_l,n_l)})\|\right\}  \nonumber \\
&   \leq (L+1) \widehat{C}_\chi^{L-1} ( 1 + \| \chi' \|_\infty + 4 \| t' \|_\infty)  \rho_1^{-L+1}  \nonumber \\
 & \quad \times \prod_{l=1}^L \left\{ \|{W}_{p_l,q_l}^{m_l,n_l}[w](\rho_1 ( r + \tilde{r}_l), \rho_1 K^{(m_l,n_l)})\| + \rho_1  {\|{W}_{p_l,q_l}^{m_l,n_l}[\partial_r w](\rho_1 ( r + \tilde{r}_l), \rho_1 K^{(m_l,n_l)})\|} \right\} \! .
\label{eq:estonVinfder}
	\end{align} 
To estimate the $\|\cdot \|_\mu$ norm we shall use the following estimate
\begin{align}
& \int_{{B}_1^{m_l+n_l}} \sup_{r \in [0,1]}
		\left\|{W}_{p_l,q_l}^{m_l,n_l}[w]( r, \rho_1 K^{(m_l,n_l)})\right\|^2
		\frac{d {K}^{(m_l,n_l)}}{|{K}^{(m_l,n_l)}|^{3 + 2 \mu}}   \nonumber   \\
	& \quad  =  \rho_1^{2 \mu(m_l +n_l)} \int_{{B}_{\rho_1}^{m_l+n_l}} \sup_{r \in [0,1]}
		\left\|{W}_{p_l,q_l}^{m_l,n_l}[w]( r,  K^{(m_l,n_l)})\right\|^2
		\frac{d {K}^{(m_l,n_l)}}{|{K}^{(m_l,n_l)}|^{3 + 2 \mu}}  \nonumber  \\
& \quad
	 \leq     \rho_1^{2 \mu(m_l +n_l)}  \frac{1}{p_l^{p_l}q_l^{q_l}} \int_{{B}_{\rho_1}^{m_l+n_l}}
		\|w_{m_l+p_l,n_l+q_l}(\cdot  , \, \cdot  \, , K^{(m_l,n_l)})\|_\mu^2
			\frac{d {K}^{(m_l,n_l)}}{|{K}^{(m_l,n_l)}|^{3 + 2 \mu}}  \nonumber \\
&  \quad \leq  \rho_1^{2 \mu(m_l+n_l)} \frac{1}{p_l^{p_l}q_l^{q_l}}  \|w_{m_l+p_l,n_l+q_l} \|_\mu^2  , \label{eq:iteratedopineq}
\end{align}
where the first equality follows by the substitution formula for integrals, the second line follows from the estimate in Lemma  \ref{lem:operatornormestimates},
and the last line follows from Fubini's theorem.
Using \eqref{eq:estonVinf} together with  \eqref{eq:iteratedopineq} we find
\begin{align}
 \|V_{{\umm,\upp,\unn,\uqq}}^{(1,\rho_1)}[t, w]\|_\mu 
	&= \left( \int_{{B}_1^{|\umm|+|\unn|}}
  \| V_{{\umm,\upp,\unn,\uqq}}^{(1,\rho_1)}[t, w]  (K^{(|\umm|,|\unn|)}) \|_{\infty}^2
  \frac{d {K}^{(|\umm|,|\unn|)}}{|{K}^{(|\umm|,|\unn|)}|^{3 + 2 \mu}}  \right)^{1/2} \nonumber  \\
  & \leq  \widehat{C}_\chi^{L-1} \rho_1^{-L+1} 
  \nonumber \\ 
		& \qquad \times \prod_{l=1}^L
		\left\{\int_{{B}_1^{m_l+n_l}} \sup_{r \in [0,1]}
		\left\|{W}_{p_l,q_l}^{m_l,n_l}[w]( r, \rho_1 K^{(m_l,n_l)})\right\|^2
		\frac{d {K}^{(m_l,n_l)}}{|{K}^{(m_l,n_l)}|^{3 + 2 \mu}} \right\}^{1/2}
\nonumber	\\ & \leq \widehat{C}_\chi^{L-1} \rho_1^{-L+1} \rho_1^{\mu(|\umm|+|\unn|)}
	\prod_{l=1}^L \left\{\frac{1}{\sqrt{p_l^{p_l}q_l^{q_l}}}  \|w_{m_l+p_l,n_l+q_l} \|_\mu \right\} \label{eq:estonVoneinf}  .
\end{align}
Similarly using \eqref{eq:estonVinfder} together with  \eqref{eq:iteratedopineq} we find
\begin{align}
 \|&\partial_r V_{{\umm,\upp,\unn,\uqq}}^{(1,\rho_1)}[t, w]  (K^{(|\umm|,|\unn|)}) \|_{\mu} \nonumber \\
	&= \left( \int_{{B}_1^{|\umm|+|\unn|}}
  \|  \partial_r V_{{\umm,\upp,\unn,\uqq}}^{(1,\rho_1)}[t, w]  (K^{(|\umm|,|\unn|)}) \|_{\infty}^2
  \frac{d {K}^{(|\umm|,|\unn|)}}{|{K}^{(|\umm|,|\unn|)}|^{3 + 2 \mu}}  \right)^{1/2} \nonumber  \\
  & \leq (L+1) \widehat{C}_\chi^{L-1} ( 1 + \| \chi' \|_\infty + 4 \| t' \|_\infty)  \rho_1^{-L+1} \nonumber \\
	&\quad \times \prod_{l=1}^L 
		\bigg\{\int\displaylimits_{{B}_1^{m_l+n_l}} \sup_{r \in [0,1]}  \left\{ \|{W}_{p_l,q_l}^{m_l,n_l}[w]( r , \rho_1 K^{(m_l,n_l)})\| 	\right.	\nonumber \\
	&	\qquad\qquad\qquad\qquad\quad\;
	\left. + \rho_1  \,{\|{W}_{p_l,q_l}^{m_l,n_l}[\partial_r w](r , \rho_1 K^{(m_l,n_l)})\|} \right\}^2 \frac{d {K}^{(m_l,n_l)}}{|{K}^{(m_l,n_l)}|^{3 + 2 \mu}} \bigg\}^{1/2}
	\nonumber \\
  & \leq (L+1) \widehat{C}_\chi^{L-1} ( 1 + \| \chi' \|_\infty + 4 \| t' \|_\infty)  \rho_1^{-L+1} \nonumber \\
 & \quad  \times \prod_{l=1}^L
		\bigg\{2  \int\displaylimits_{{B}_1^{m_l+n_l}}  \bigg( \sup_{r \in [0,1]}   \|{W}_{p_l,q_l}^{m_l,n_l}[w]( r , \rho_1 K^{(m_l,n_l)})\|^2  
\nonumber \\
&	\qquad\qquad\qquad\;
	+ \rho_1    \sup_{r \in [0,1]} {\|{W}_{p_l,q_l}^{m_l,n_l}
	[\partial_r w](r , \rho_1 K^{(m_l,n_l)})\|}^2 \bigg)
	\frac{d {K}^{(m_l,n_l)}}{|{K}^{(m_l,n_l)}|^{3 + 2 \mu}}  \bigg\}^{1/2}
\nonumber \\
& \leq ( L+1) \widehat{C}_\chi^{L-1}
	( 1 + \| \chi' \|_\infty + 4 \| t' \|_\infty)\rho_1^{-L+1}
	\rho_1^{\mu(|\umm|+|\unn|)} 
	\prod_{l=1}^L \left\{\frac{\sqrt{2}}{\sqrt{p_l^{p_l}q_l^{q_l}}}
	\|w_{m_l+p_l,n_l+q_l} \|_\mu^\# \right\}  .\label{eq:estonVoneinf2}
\end{align}
Adding Inequalities \eqref{eq:estonVoneinf} and \eqref{eq:estonVoneinf2}
 establishes the  estimate  in \eqref{eq:basicestseond}. 
\end{proof}

The proof of  Theorem  \ref{thm:ballprofforsecond} is  similar to the proof of Theorem 3.8. given in \cite{BCFS}.

\vspace{0.2cm}
\enlargethispage{0.5cm}
\noindent
{\it Proof of Theorem  \ref{thm:ballprofforsecond}.}
The operators $H(t)$ and $H(w)$ are a Feshbach pair for $ \boldsymbol{\chib}^{(1)}_{\rho_1} $
by Theorem  \ref{prop:feshsecond}.
First  we use  Lemma \ref{lem:approx:V} to  estimate \eqref{eq:defofwmnschlange(1)}. 
We shall write  $C_t :=  1 + 2 \|\chi'\|_\infty + 8 \tau_0.$
Since  $C_\chi :=  \sqrt{2} \hat{C}_\chi$  and $\hat{C}_\chi \geq 1$, we have  $ 2^{L/2} \hat{C}_\chi^{L-1} \leq C_\chi^L$. Moreover, we   use that $\binom{m+p}{ p} \leq 2^{m+p}$.
With this we find  for $M+ N \geq 1$
\begin{align}
 &\|w^{(2,{\rho_{1}})}_{M,N}\|_\mu^\# \nonumber \\
& \leq   \sum_{L=1}^\infty
\sum_{\substack{ {\umm,\upp,\unn,\uqq} \in \N_0^{L} \\   |\umm|=M, \, |\unn|=N  } }
\prod_{l=1}^L \left\{\binom{ m_l + p_l}{ p_l} \binom{ n_l + q_l}{ q_l } \right\} 
(L+2) C_\chi^{L} C_t \rho_1^{(1+\mu)(M+N)-L}
\frac{ \| w_{m_l+p_l,n_l+q_l} \|_\mu^\#  }{ \sqrt{ p_l^{p_l} q_l^{q_l}}} \nonumber \\
& \leq   \sum_{L=1}^\infty (L+2)   \left(\frac{C_\chi}{\rho_1}\right)^L C_t
\left(2\rho_1^{(1+\mu)}\right)^{M+N} 
\sum_{\substack{ {\umm,\upp,\unn,\uqq} \in \N_0^{L}\\   |\umm|=M, \, |\unn|=N  } }
\prod_{l=1}^L \left\{\left(\frac{2}{\sqrt{p_l}}\right)^{p_l}
\left(\frac{2}{\sqrt{q_l}}\right)^{q_l} \| w_{m_l+p_l,n_l+q_l} \|_\mu^\# \right\}.
\end{align} 
Inserting this in $\|\cdot\|_{\mu,\xi}^\#$ we obtain the following bound 
\begin{align} \label{3.100}
 \|&(w^{(2,{\rho_{1}})}_{M,N})_{M+ N  \geq 1}\|_{\mu,\xi}^\# \nonumber\\
& \leq  C_t  \sum_{M+N \geq 1}{\xi^{-(M+N)}
\|w^{(2,{\rho_{1}})}_{M,N}\|_\mu^\#} \nonumber  \\
& \leq  2 C_t  \rho_{1}^{(1+\mu)}\sum_{L = 1}^\infty(L+2)
\left(\frac{C_\chi}{\rho_1}\right)^L \sum_{M+ N \geq 1}
 \nonumber  \\
 & \qquad \quad  
	\sum_{\substack{ {\umm,\upp,\unn,\uqq} \in \N_0^{L} \\  |\umm|=M, \, |\unn|=N } }
	\prod_{l=1}^L \left\{\left(\frac{2 \xi }{\sqrt{p_l}}\right)^{p_l}
	\left(\frac{2 \xi}{\sqrt{q_l}}\right)^{q_l} \xi^{-(m_l+p_l+n_l+q_l)} \| w_{m_l+p_l,n_l+q_l} \|_\mu^\# \right\}
	\nonumber \\
& \leq  2 C_t  \rho_1^{(1+\mu)}\sum_{L = 1}^\infty(L+2)
	\left(\frac{C_\chi}{\rho_1}\right)^L \nonumber   \\
	& \qquad \quad \left\{\sum_{m,n,p,q \in \N_0 }\left(\frac{2 \xi }{\sqrt{p}}\right)^{p}
	\left(\frac{2 \xi}{\sqrt{q}}\right)^{q} \xi^{-(m+p+n+q)} \| w_{m+p,n+q} \|_\mu^\# \right\}^L \nonumber   \\
& \leq  2 C_t  \rho_1^{(1+\mu)}\sum_{L = 1}^\infty(L+2)
	\left(\frac{C_\chi}{\rho_1}\right)^L \nonumber   \\
	& \qquad \quad \left\{\sum_{m,n \in \N_0}
	\left(\sum_{p=0}^m\left(\frac{2\xi}{\sqrt{p}}\right)^p\right)
	\left(\sum_{q=0}^m\left(\frac{2\xi}{\sqrt{q}}\right)^q\right)
	\xi^{-(m+n)}\|w_{m,n}\|_\mu^\# \right\}^L \nonumber   \\
& \leq 2 C_t  \rho_1^{(1+\mu)}\sum_{L = 1}^\infty(L+2)
	\left(\frac{C_\chi}{\rho_1}\right)^L 4^{L} \left( \|w \|_{\mu,\xi}^\#\right)^L 	, 
 \end{align}
 where  in the second last inequality we used a substitution of summation variables and
 in the last inequality we used that $\sum_{p=0}^\infty\left(\frac{2\xi}{\sqrt{p}}\right)^p \leq \sum_{p=0}^\infty\left(2\xi\right)^p = \frac{1}{1-2\xi}  \leq 2$ since $0 < \xi \leq 1/4$.
By  the assumptions of the theorem we have
$$ \|w \|_{\mu,\xi}^\# \leq  \gamma .  $$
Inserting this into \eqref{3.100} we find
\begin{align} \label{3.100v2}
 \|(w^{(2,{\rho_{1}})}_{M,N})_{M+N \geq 1}\|_{\mu,\xi}^\#
   &\leq 2  C_t  \rho_1^{(1+\mu)}\sum_{L = 1}^\infty(L+2)
   \left(\frac{4 C_\chi\gamma}{\rho_1} \right)^L \nonumber  \\
   &\leq 24  C_t  \rho_1^{(1+\mu)} \frac{C_\chi \gamma}{\rho_1}
   	\left(1 - \frac{4 C_\chi \gamma}{\rho_1}\right)^{-2} ,
\end{align}
where we used that by assumption
$$
	0 \leq \frac{4 C_\chi \gamma}{\rho_1} < 1,
$$
and moreover  we used that
$\sum_{L=1}^\infty (L+2) a^L = \sum_{L=3}^\infty  L  a^{L-2} \!
	= a^{-1} \frac{d}{da} \sum_{L=3}^\infty   a^{L} = a^{-1} \frac{d}{da} \frac{a^3}{1-a} \leq \frac{3a}{(1-a)^2}$
for $a \in (0,1)$. \\
It remains to  estimate  \eqref{w:1:g:00}. To this end we
recall that for $\umm = \unn = \underline{0}$ we see from Lemma \ref{lem:approx:V} that
\begin{equation}
\rho_1^{- 1} \| V_{{\underline{0},\upp,\underline{0},\uqq}}^{(1,\rho_1)}[t, w]  \|_\mu^\# \leq (L+2) C_\chi^{L} C_t \rho_1^{-L} \prod_{l=1}^L
\frac{ \| w_{p_l,q_l} \|_\mu^\# }{ \sqrt{ p_l^{p_l} q_l^{q_l}}}.
\end{equation}
Using this we find for the derivative
\begin{align}
	\sup_{r\in[0,1]} &| \partial_r \hat{w}^{(2,{\rho_{1}})}_{0,0}(r) - 1 | \nonumber \\
	 &= \sup_{r\in[0,1]} \bigg| \rho_1^{-1} \partial_r t(\rho_1 r) + \rho_1^{-1}
	 \Big(\sum_{L=1}^\infty (-1)^{L-1} \sum_{\substack{ \upp,\uqq \in \N_0^{L}} }
	 \partial_r V_{{\underline{0},\upp,\underline{0},\uqq}}^{(1,\rho_1)}[t, w](r) \Big) - 1 \bigg| \nonumber \\
	 &\leq   \|t' - 1\|_\infty
		+  \rho_1^{-1} \sum_{L=1}^\infty \sum_{\substack{ \upp,\uqq \in \N_0^{L}} } \| V_{{\underline{0},\upp,\underline{0},\uqq}}^{(1,\rho_1)}[t, w]  \|_\mu^\#  \nonumber \\
	 &\leq   \tau_1
		+\sum_{L=1}^\infty (L+2) C_\chi^{L} C_t \rho_1^{-L} \sum_{\substack{ \upp,\uqq \in \N_0^{L}} } \prod_{l=1}^L \frac{ \| w_{p_l,q_l} \|_\mu^\# }{ \sqrt{ p_l^{p_l} q_l^{q_l}}} \nonumber \\
	 &\leq   \tau_1
		+ \sum_{L=1}^\infty (L+2) C_\chi^{L}C_t  \rho_1^{-L}
			\left( \sum_{p,q\in \N_0} \frac{ \| w_{p,q} \|_\mu^\# }{\sqrt{ p^{p} q^{q}}}\right)^L \nonumber \\
	 &\leq   \tau_1
		+ \sum_{L=1}^\infty  (L+2) C_\chi^{L} C_t  \rho_1^{-L}
			\big[   \| (w_{m,n})_{m+n \geq 0} \|_{\mu,\xi}^\#\big]^L \nonumber \\
	 &\leq   \tau_1
		+  \sum_{L=1}^\infty (L+2) C_t  \left( \frac{ C_\chi  \gamma}{\rho_1}     \right)^L \nonumber  \\
 &\leq \rho_1^{-1} \tau_1   +  3 C_t \frac{ C_\chi  \gamma}{\rho_1}  (1 - \frac{ C_\chi  \gamma}{\rho_1})^{-2} ,    \label{3.100v3}
\end{align}
where we used that by assumption
\begin{equation} \label{eq:finalestbound}
0 \leq \frac{ C_\chi  \gamma}{\rho_1} <  1 .
\end{equation}
Analogously we estimate
\begin{align}
|\hat{w}^{(2,{\rho_{1}})}_{0,0}(0) -  \rho_1^{-1} t(0)  |
	 &\leq 	| \rho_1^{-1}   \sum_{L=1}^\infty (-1)^{L-1}
		\sum_{\substack{ \upp,\uqq \in \N_0^{L}} }
			V_{{\underline{0},\upp,\underline{0},\uqq}}^{(1,\rho_1)}[t, w](0) | \nonumber  \\
	&\leq \rho_1^{-1} \sum_{L=1}^\infty \sum_{\substack{ \upp,\uqq \in \N_0^{L}} }
	         \| V_{{\underline{0},\upp,\underline{0},\uqq}}^{(1,\rho_1)}[t, w] \|_\infty  \nonumber  \\
&\leq 3 C_t  \frac{ C_\chi  \gamma}{\rho_1}  (1 - \frac{ C_\chi  \gamma}{\rho_1})^{-2} , \label{3.100v4}
\end{align}
provided  \eqref{eq:finalestbound} holds. The claim now follows from Equations \eqref{3.100v2}, \eqref{3.100v3}  and \eqref{3.100v4}.
\qed

\section{Proof of the Main Theorem}
 \label{sec:prooofmain}

Let $B_R$ denote the closed ball with radius $R$.

\begin{theorem}[Analyticity Theorem of Griesemer-Hasler \cite{GriHas09}] \label{thm:griehas} Let  $\mu > 0$, and  let   $\rho_{\rm GH} \in (0,1)$ be sufficiently small.
Then for  $\xi \in (0,1)$ sufficiently small    there exist positive constants $\alpha_0$, $\beta_0$
 and $\gamma_0$ such that the following holds. \\
Let     $\VV $ be an open subset of  $\C$ and let $E_{\rm at} : \VV \to \C$ be an analytic function.
 Let $\UU$ be an open subset of $\C \times \C$  such that for all $s \in \VV$ we have
\begin{equation}  \label{eq:condonuu}
\{ s \} \times B_{\rho_{\rm GH}}(E_{\rm at}(s)) \subset \UU    .
\end{equation}
Suppose  $H(\cdot,\cdot)$ is an $\mathcal{L}(\FF)$ valued analytic function on  $\UU$,
such that  for all $(s,\zeta) \in \UU$
 $$
 H(s,\zeta) - (E_{\rm at}(s) - \zeta ) \in \mathcal{B}^{[1]}(\alpha_0,\beta_0,\gamma_0) .
 $$
 Then there
 exist  analytic functions $\zeta_\infty : \VV \to  B_{\rho_{\rm GH}}(s)$
 and   $\psi_\infty : \VV \to \HH_{\rm red}$, nowhere vanishing, such that    for all $s \in \VV$
 $$
 H(s,\zeta_\infty(s))\psi_\infty(s) = 0 .
 $$
 If furthermore $H(s,\zeta)^* = H(\overline{s},\overline{\zeta})$ on $\mathcal{U}$,  then for real $s \in \VV$ the operator  $H(s,\lambda)$ has a bounded inverse  for all  $\lambda  \in  (E_{\rm at}(s)-\rho_{GH} , \zeta_\infty(s))$.
  \end{theorem}

The proof of this theorem follows directly  from the proof of Theorem 1 in \cite[pp.\ 610-611]{GriHas09}.
The only difference being that the theorem given above does not
involve the initial Feshbach step and has therefore a shorter proof.
In the application of Theorem  \ref{thm:griehas}, which we have in mind,
the  parameter $s$ in the theorem will be played by the coupling constant $g$.
Furthermore the energy cutoff $\rho_{\rm GH}$ will be an order one quantity, i.e., independent of the coupling constant.
Now we are ready to prove the main result of this paper.

\vspace{0.2cm}

\noindent
{\it Proof of Theorem \ref{thm:main}.} Fix $\mu > 0$.  Let   $\rho_{\rm GH} \in (0,1/2]$ and $\xi\in (0,1/4]$
by chosen sufficiently small such that the
 assertion of Theorem   \ref{thm:griehas} holds.
 The idea will be to choose energy cutoffs, $\rho_0$ and $\rho_1$,  at the first and second Feshbach step
 so that for $g$ in a sectorial region of an annulus, we can apply Theorem \ref{thm:griehas}.
 As we let   the outer and inner radius of the annulus tend to zero we will obtain the desired result.
 $$
 $$
Choose $\epsilon \in (0,\mu)$ and  $\alpha \in(0,\min(\mu-\epsilon,1))$. For  $\rho_0 > 0$ we define
\begin{equation} \label{eq:proofofmain0}
\rho_1 := \rho_0^{1 + 2\epsilon + \alpha}  .
\end{equation}
We shall assume that
\begin{equation} \label{eq:proofofmain2}
\rho_0 \in (0,1/4)   .
\end{equation}
We consider the following sectorial region  of an annulus, determined by the conditions  $g \in S_{\delta_0}$ and
\begin{equation} \label{eq:proofofmain1}
 c_{\delta_0}^{-1/2}  \rho_0^{ 1 + \epsilon +   \frac{\alpha}{2}  } <   |g| <  \min( \rho_0^{ 1 + \epsilon } , (8 \|Z_{\rm at}\| + 4 c_{\delta_0})^{-1} \rho_0^{1/2} ,(  8 \xi^{-1} C_F  \| G \|_\mu    )^{-1} \rho_0^{1/2})   .
\end{equation}
In view of  the upper bound in   \eqref{eq:proofofmain1},
we conclude from    Theorem  \ref{thm:inimain1} (Banach Space Estimate for 1st  Feshbach Operator) that there exists a finite constant $C^{(1)}$ such that for
$\rho_0$ satisfying \eqref{eq:proofofmain2}  and   all $g \in S_{\delta_0}$ with  \eqref{eq:proofofmain1}
and $z \in D_{\rho_0/2}(\epsilon_{\rm at})$ we have
\begin{equation} \label{eq:proofmainbanach1st}
H_g^{(1,{\rho_{0}})}(z) - \rho_0^{-1} (\epsilon_{\rm at}  - z )  \in  \mathcal{B}^{[d_0]}(  C^{(1)} \rho_0^{\frac{1}{2} + \epsilon}   ,  C^{(1)} \rho_0^{\frac{1}{2} + \epsilon}     ,   C^{(1)}  \rho_0^{\mu + 1 + \epsilon}  )  .
\end{equation}

Next we want to use  Theorem \ref{prop:feshsecond} (Feshbach Pair Criterion  for 2nd  Iteration).
Observe that  by  \eqref{eq:proofofmain2} the assumptions on the $\rho$'s and by \eqref{eq:proofofmain2}  the assumptions \eqref{eq:secondfeshestm1}  and \eqref{eq:secondfeshest2}  of  Theorem   \ref{prop:feshsecond}
are satisfied for $g \in S_{\delta_0}$ with \eqref{eq:proofofmain1}. To apply the theorem we consider the decomposition defined in  \eqref{eq:decompsecond}
where
$$
{H}_g^{(1,\rho_0)}(z) = T_{g }^{(1,\rho_0)}(z)     +   W_{g }^{(1,\rho_0)}(z)
$$
By Lemma  \ref{norm:w:(0b)} (Free Approx to  1st  Feshbach) we see that for the difference to  the free approximation \eqref{eq:tfreeapprox} there is  some constant $C$
such that
\begin{align} \label{eq:approxestproofmain}
 & \| t^{(1,{\rho_{0}})}_{g}(z)- t_{g, {\rm free}}^{(1,{\rho_{0}})}(z) \|
	 \leq C \rho_0^{2  +  2\epsilon}  , \\
& \| {P}_{\rm at}^{(2)}  w_{g,0,0}^{(1,\rho_0)} \overline{P}_{\rm at}^{(2)} +  \overline{P}_{\rm at}^{(2)}  w_{g,0,0}^{(1,\rho_0)} {P}_{\rm at}^{(2)}  \|
	 \leq C \rho_0^{2  +  2\epsilon}    ,\label{eq:approxestproofmain2}
\end{align}
where we used the upper bound in \eqref{eq:proofofmain1}.
By \eqref{eq:proofofmain0} we see that the right hand side of  \eqref{eq:approxestproofmain}
is less than $\rho_1/4$ provided $\rho_0$ is sufficiently small. Thus by Neumann's theorem
and Lemma  \ref{T-invert} we see  that $T_{g }^{(1,\rho)}(z)$  is invertible on  $\ran \boldsymbol{\chib}^{(1)}_{\rho_1} $  and    $$  \| ( H_{0,0}(t_g^{(1,\rho_0)}) \upharpoonright \ran \boldsymbol{\chib}^{(1)}_{\rho_1} )^{-1}  \| \leq \frac{8}{\rho_1}. $$
Furthermore, we see from \eqref{eq:approxestproofmain2}  and \eqref{eq:proofmainbanach1st}  that
$$
\| w^{(1,{\rho_{0}})}_{g,{\rm int}}(z)  \|_{\mu,\xi}^\#
		\leq C \rho_0^{2  +  2\epsilon} +   C^{(1)}  \rho_0^{\mu+ 1+ \epsilon} =: \gamma  .
$$
We note that
$$
\frac{\gamma}{\rho_1}= C \rho_0^{1 - \alpha}  +    C^{(1)} \rho_0^{\mu - \epsilon - \alpha }  ,
$$
 tends to zero for $\rho_0 \to 0$.
In view  of   \eqref{eq:injective} we see that    condition  \eqref{eq:condonsecfesh} of Theorem \ref{prop:feshsecond}
is satisfied for $\rho_0$ small.
Thus
 $ T_{g}^{(1,\rho_0)}(z)$ and  $W_{g}^{(1,{\rho_{0}})}(z)$  are  a Feshbach pair for $\boldsymbol{\chi}^{(1)}_{\rho_1}$,
provided $g  \in S_{\delta_0}$,  \eqref{eq:proofofmain1} holds,
and  $z \in   D_{\rho_0 \rho_1/2}(\epsilon_{\rm at} + g^2 \epsilon_{\rm at}^{(2)})$.
It now  follows from Theorem \ref{thm:ballprofforsecond} (Banach Space Estimate for 2nd Feshbach Operator) that for some constant $C^{(2)}$ we have
\begin{align}
 S_{\rho_1} ( F_{\boldsymbol{\chi}^{(1)}_{\rho_1}}( T_{g}^{(1,{\rho_{0}})}(z) ,  W_{g}^{(1,{\rho_{0}})}(z)  )  )&  -  \rho_1^{-1} \rho_0^{-1}  (\epsilon_{\rm at} - z + g^2  \epsilon_{\rm at}^{(2)} )
\label{eq:feshfinal}  \\
& \in
 \mathcal{B}^{[1]}( C^{(2)}  \frac{\gamma}{\rho_1}  ,   C^{(2)}  \rho_0^{\frac{1}{2} + \epsilon}   + C^{(2)} \frac{\gamma}{\rho_1}  ,
  C^{(2)}  \rho_1^\mu \gamma  ) \nonumber
 \end{align}
 for all $z \in   D_{\rho_0 \rho_1/2}(\epsilon_{\rm at} + g^2 \epsilon_{\rm at}^{(2)})$ and  $g  \in S_{\delta_0}$ with \eqref{eq:proofofmain1}.

Now we want to apply  Theorem  \ref{thm:griehas} (Analyticity Theorem of Griesemer-Hasler). To this
end we   express the Hamiltonian in terms of the  variables   $ \zeta  = \rho_0^{-1}\rho_1^{-1} z$ and $s = g$.
We define the function  $E_{\rm at}(s) =  \rho_1^{-1} \rho_0^{-1} ( \epsilon_{\rm at} + s^2 \epsilon_{\rm at}^{(2)})$ on $\C$ and the function
$$
 H(s,\zeta) =  S_{\rho_1} ( F_{\boldsymbol{\chi}^{(1)}_{\rho_1}}( T_{s }^{(1,{\rho_{0}})}(\rho_0\rho_1 \zeta) ,  W_{s}^{(1,{\rho_{0}})}(\rho_0 \rho_1 \zeta)  )  )
$$
for $(s,\zeta) \in \mathcal{U} := \bigcup_{ g \in S_{\delta_0} :  \eqref{eq:proofofmain1}} D_{1/2}( E_{\rm at}(g))$.
Since we  expressed   $H(s,\zeta)$ in terms of  uniformly convergent Neumann series (Theorems \ref{pro:firstbound} and \ref{prop:feshsecond}), one concludes that $H(s,\zeta)$ is
 jointly  analytic on  $\mathcal{U}$ and that the  conjugation property $H(s,\zeta)^* = H(\overline{s},\overline{\zeta})$ holds, since each term in the convergent expansion has that property.
We conclude now  by Theorem  \ref{thm:griehas} (Analyticity Theorem of Griesemer-Hasler) (if  $\rho_{GH}$ is less than $1/2$, then    \eqref{eq:condonuu} holds) that there exist  analytic functions
$$
   g \mapsto \zeta_\infty(g)  , \quad  g \mapsto \psi_\infty(g)
$$
for $g \in S_{\delta_0}$ with  \eqref{eq:proofofmain1}
such that
$$
H(g, \zeta_\infty(g)  ) \psi_\infty(g)     = 0 .
$$ 
and moreover $\zeta_\infty(g) \in D_{1/2}(E_{\rm at}(g))$. Expressed in  terms of the original  variables we obtain
 from the isospectrality of the Feshbach map, that
$$
E_g :=  \rho_0 \rho_1 \zeta_\infty(g)    \quad \text{and }     \quad  \psi_g :=  Q^{(0,{\rho_{0}})}_g(E_g)  Q^{(1,{\rho_{1}})}_g(E_g)  \psi_\infty(g)  ,
$$
where
\begin{align*}
 Q^{(0,{\rho_{0}})}_g(z) & :=   \boldsymbol{\chi}_{\rho_0}^{(0)} -  \boldsymbol{\chib}_{\rho_0}^{(0)} \left(H_0- z +  \boldsymbol{\chib}_{\rho_0}^{(0)} g W  \boldsymbol{\chib}_{\rho_0}^{(0)} \right)^{-1}g W \boldsymbol{\chi}_{\rho_0}^{(0)}  \\
 Q^{(1,{\rho_{1}})}_g(z)  & :=   \boldsymbol{\chi}^{(1)}_{\rho_1} -  \boldsymbol{\chib}^{(1)}_{\rho_1} \left( T_{g}^{(1,{\rho_{0}})}(z) +  \boldsymbol{\chib}^{(1)}_{\rho_1} W_{g}^{(1,{\rho_{0}})}(z)  \boldsymbol{\chib}^{(1)}_{\rho_1} \right)^{-1}  W_g^{(1,\rho_0)} \boldsymbol{\chi}^{(1)}_{\rho_1} ,
\end{align*}
are eigenvalue and eigenvector of $H_g$.
It now follows that $E_g$ and $\psi_g$ are analytic functions of $g \in S_{\delta_0}$ with  \eqref{eq:proofofmain1}
since  $Q^{(0,{\rho_{0}})}_g(z)$ and   $Q^{(1,{\rho_{1}})}_g(z)$ are analytic functions of $g$ and $z$, as they are given by convergent expansions
of jointly analytic functions. Furthermore in terms of the original spectral parameter we have  $E_g \in  D_{\rho_0 \rho_1/2}(\epsilon_{\rm at} + g^2 \epsilon_{\rm at}^{(2)})$,
which implies   \eqref{expansionmain}. The conjugation property of $H(s,\zeta)$, the last statement  in Theorem \ref{thm:griehas}, and the isospectrality property of the  Feshbach map
imply that $E_g$ is the ground state energy of $H_g$.
As we take $\rho_0$ to  zero, the theorem follows.
\qed

\begin{remark}
{\rm We note that the choice for $\rho_0$ and $\rho_1$ in the proof corresponds to $\rho_0$ being larger than $|g|$ and $\rho_1$ being smaller than $|g|$ such that their product is smaller than $|g|^2$.}
\end{remark}

\section*{Acknowledgements}

D.H. wants to thank I. Herbst and M. Westrich for valuable conversations. Furthermore we want to thank G. Br\"aunlich for
helpful discussions.

\appendix

\section{Elementary Estimates }
\label{sec:appA}

\begin{lemma} We have the following estimate
\begin{equation}  \label{eq:projineq1}
\| a(G) 1_{H_f \leq r }  \|  \leq  \left( \int_{|k| \leq r }   \frac{\| G(k) \|^2}{|k|}  dk \right)^{1/2}  r^{1/2}
\end{equation}
\end{lemma}
\begin{proof}
This follows from the following estimate
\begin{align*}
\| a(G) 1_{H_f \leq r } \| &=   \| a(G 1_{|k| \leq r } ) 1_{H_f \leq r } \|
=  \| a(G 1_{|k| \leq r } ) H_f^{-1/2} H_f^{1/2}  1_{H_f \leq r } \|  \\
& \leq    \| a(G 1_{|k| \leq r } ) H_f^{-1/2} \| \|  H_f^{1/2}  1_{H_f \leq r } \| \\
& \leq  \left( \int_{|k| \leq r }   \frac{\| G(k) \|^2}{|k|}  dk \right)^{1/2} r^{1/2}.
\tag* \qedhere
\end{align*}
\end{proof}

In \cite{GriHas09} the following Lemma is shown.

\begin{lemma} For  $G \in L^{2}(\R^3 \times \Z_2 ; \mathcal{L}(\HH_{\rm at}))$ we have
\begin{align} \label{eq:estoncrea}
\| a(G)H_f^{-1/2} \| & \leq \| \omega^{-1/2} G \|  , \\      \| a^*(G)(H_f+1)^{-1/2} \| & \leq \|  (\omega^{-1} +1)^{1/2} G \|   \nonumber                             .
\end{align}
\end{lemma}

\begin{proof}
We will use the notation of eq. \eqref{eq:easyNotation}.
We set $N := \int a^*(k)  a(k)  dk $
and let $\psi \in 1_{N \leq n} \HH$ for some $n \in \N$.
To prove the first inequality we estimate
\begin{align*}
\| a(G) \psi \| &\leq
\int \| G(k) a(k) \psi \| dk =  \int \| G(k)|k|^{-1/2} |k|^{1/2}  a(k) \psi \| dk  \\
&\leq   \left( \int  |k| \| a(k) \psi \|^2 dk \right)^{1/2}
\left(\int  |k|^{-1} \| G(k) \|^2 dk \right)^{1/2} \\
&=
\left(\int  |k|^{-1} \| G(k ) \|^2 dk \right)^{1/2}  \| H_f^{1/2} \psi \|.
\end{align*}
To prove the second we use the commutation relations
\begin{align*}
\| a^*(G) \psi \|^2 &= \inn{ a^*(G) \psi ,  a^*(G) \psi } =  \inn{  \psi , a(G) a^*(G) \psi }\\  &=
\inn{  \psi , ( a^*(G) a(G) + \int \| G(k) \|^2 dk  ) \psi } \\
&\leq
\left(  \int  |k|^{-1} \| G(k) \|^2 dk \right)  \| H_f^{1/2} \psi \|^2  +
 \int \| G(k) \|^2 dk   \| \psi \|^2 . \tag* \qedhere
\end{align*}
\end{proof}

 The following lemma states the well
known pull-through formula.
For a proof see for example \cite{BacFroSig98-2,HasHer11-1}.
\begin{lemma} \label{lem:pullthrough}
Let $f : \R_+ \to \C$ be a bounded measurable function. Then for all $K \in \R^3 \times \Z_2$
$$
f(H_f) a^*(K) = a^*(K) f(H_f + \omega(K) ) , \quad a(K) f(H_f) = f(H_f + \omega(K) ) a(K) .
$$
\end{lemma}

\section{Field Operators Associated to  Integral Kernels}
\label{sec:appfielddef}

In what follows we shall give a precise meaning to field operators defined by integral kernels. 
Let ${X} := \R^3 \times \Z_2$.
For $\psi \in \FF$ having finitely many particles we have
\beqn \label{eq:defofa}
\left[ a(K_1) \cdots a(K_m) \psi \right]_n(K_{m+1},...,K_{m+n}) = \sqrt{\frac{(m+n)!}{n!}} \psi_{m+n}(K_{1},...,K_{m+n}) ,
\eeqn
for all $K_1,...,K_{m+n} \in X$, and
using Fubini's theorem it is elementary to see that the vector valued map
 $(K_1,...,K_m) \mapsto a(K_1) \cdots a(K_m) \psi$ is
an element of $L^2(X^m; \FF)$.
For measurable functions $w_{m,n}$ on $ {(\R^3 \times \Z_2)}^{n+m}$ with values
in the linear operators of $\HH_{\rm at}$
we define  the sesquilinear form
\begin{equation*} 
\int_{{(\R^3 \times \Z_2)}^{m+n}} \frac{ dK^{(m,n)}}{|K^{(m,n)}|^{1/2}}
 \left\langle a(K^{(m)}) \varphi ,w_{m,n}( K^{(m,n)}) a(\widetilde{K}^{(n)}) \psi \right \rangle ,
\end{equation*}
defined  for all $\varphi$ and $\psi$ in $\HH$, for which the integrand on the right hand side is integrable.
 This yields by the representation
 theorem a densely defined linear operator, which can easily be shown to be closable. We denote by $H^{(0)}(w_{m,n})$ the
 closure of the operator. By adjusting the notation this also provides  a precise definition of the field operators $a(G)$ and $a^*(G)$.
 To define field operators which depend on the free field energy we consider
measurable functions $w_{m,n}$ on $\R_+ \times {(\R^3 \times \Z_2)}^{n+m}$ with values
in the linear operators of $\HH_{\rm at}$.
To such a function we associate the sesquilinear form
\begin{equation*}
q_{w_{m,n}}(\varphi,\psi) \!:=\!\! \int_{{(\R^3 \times \Z_2)}^{m+n}} \!\frac{ dK^{(m,n)}}{|K^{(m,n)}|^{1/2}}\!
 \left\langle \!a(K^{(m)}) \varphi ,w_{m,n}(H_f, K^{(m,n)}) a(\widetilde{K}^{(n)}) \psi \right \rangle \!,
\end{equation*}
defined  for all $\varphi$ and $\psi$ in $\HH$, for which the integrand on the right hand side is integrable.
If the integral kernel decays sufficiently fast as a function of the free field energy,
the sesquilinear form defines a bounded operator. For this  we can use the following lemma and the identification in \eqref{eq:identbsp}.
\begin{lemma} \label{kernelopestimate}  For measurable $w : X^{m+n} \to C_\infty[0,\infty)$, define
\begin{align*}
 \| w_{m,n} \|_\sharp^2 
 := &
\int_{{X}^{m+n}}\!\! \frac{d K^{(m,n)}}{|K^{(m,n)}|^2} \\
&\; \times \sup_{r \geq 0} \left[ \|w_{m,n}(r,K^{(m,n)}) \|^2 \prod_{l=1}^m \left\{ r + \Sigma[K^{(l)}] \right\}
 \prod_{\widetilde{l}=1}^n \left\{ r + \Sigma[\widetilde{K}^{(\widetilde{l})}] \right\} \right] .
\end{align*}
Then for all $\varphi, \psi \in \HH$ with finitely many particles
\begin{equation} \label{eq:defofH}
| q_{w_{m,n}}(\varphi,\psi) | \leq \| w_{m,n} \|_\sharp \| \varphi \| \| \psi \| .
\end{equation}
In particular if $\| w_{m,n} \|_\sharp  < \infty$
 the form $q_{w_{m,n}}$ determines uniquely a bounded
linear operator  ${H}_{m,n}(w_{m,n})$ such that
$$
q_{w_{m,n}}(\varphi,\psi ) = \langle \varphi,{H}_{m,n}(w_{m,n}) \psi \rangle ,
$$
for all $\varphi, \psi$ in $\HH$. Moreover,
$\| {H}_{m,n}(w_{m,n}) \|_{} \leq \| w_{m,n} \|_\sharp$.
\end{lemma}
\begin{proof} We set $P[K^{(n)}] := \prod_{l=1}^n ( H_f + \Sigma[K^l])^{1/2}$ and insert 1's to obtain the trivial identity
\begin{align*}
 | q_{w_{m,n}}(\varphi,\psi) | 
 &= \Bigg| \int_{{X}^{m+n}} \frac{d K^{(m,n)}}{|K^{(m,n)}|}
 \Big\langle
 P[K^{(m)}] P[K^{(m)}]^{-1} |K^{(m)}|^{1/2}  \\
 & \qquad \quad \times a(K^{(m)}) \varphi , w_{m,n}(H_f ,K^{(m,n)})
P[\widetilde{K}^{(n)}] P[\widetilde{K}^{(n)}]^{-1} | \widetilde{K}^{(n)}|^{1/2} a(\widetilde{K}^{(n)}) \psi \Big\rangle \Bigg| .
\end{align*}
The lemma now follows using the Cauchy-Schwarz
inequality and the following well known identity for $n \geq 1$ and $\phi \in \FF$,
\begin{equation}
\int_{{X}^n}
d K^{(n)} | K^{(n)} | \left\| \prod_{l=1}^n \left[ H_f + \Sigma[K^{(l)}] \right]^{-1/2} a(K^{(n)}) \phi \right\|^2 = \| P_\Omega^\perp \phi \|^2 \label{eq:trivialA}  ,
\end{equation}
where
$P_\Omega^\perp := | \Omega \rangle \langle \Omega |$.
A proof of \eqref{eq:trivialA}  can for example be found in \cite{HasHer11-1} Appendix A.
The last statement of the lemma follows from the first and the representation theorem.
\end{proof}

\section{ Generalized Wick Theorem}
\label{sec:appB}

For $m,n \in \N_0$ let  $\mathcal{M}_{m,n}$ denote the space of measurable functions on $\R_+ \times (\R^3 \times \Z_2)^{m+n}$ with values
in the linear operators of $\HH_{\rm at}$. Let
$$
\mathcal{M} = \bigoplus_{m+n=1} \mathcal{M}_{m,n}.
$$
Then using the notation introduced in Section \ref{sec:firstStep} we define for
$w \in \mathcal{M}$
\begin{align}\label{eq:WforFirst}
	W[w] := \sum_{m+n =1} {H}_{m,n}^{^{(0)}}(w) .
\end{align} 
Moreover for  $w \in \WW_{\xi}^{[d]}$ we define, similar to \eqref{eq:opHw},
\begin{align}\label{eq:WforSecond}
W[w] := \sum_{m,n \in \N_0} H_{m,n}(w) .
\end{align}
The following Theorem is from \cite{BacFroSig98-2}. It is a generalization of Wick's Theorem.
\begin{theorem} \label{thm:wicktheorem} Let $w \in \mathcal{M}$ or
$w \in \WW_{\xi}^{[d]}$ and let
$F_0,F_1,...,F_L \in \mathcal{M}_{0,0}$ resp. $F_0,F_1,...,F_L \!\in \WW_{0,0}^{[d]}$.
Then as a formal identity
$$
F_0(H_f) W[w] F_1(H_f) W[w] \cdots W[w] F_{L-1}(H_f) W[w] F_L(H_f)
	= H( \widetilde{w}^{({\rm sym})} ) ,
$$
where $\widetilde{w}^{({\rm sym})}$ is the symmetrization w.r.t. $k^{(M)}$ and
$\widetilde{k}^{(N)}$ of
\begin{align}
\widetilde{w}_{M,N}(r;K^{(M,N)})  \nonumber & =
\sum_{\substack{ m_1 + \ldots + m_L = M \\ n_1 + \ldots + n_L=N }}
\sum_{\substack{ p_1, q_1\ldots,p_L,q_L: \\ m_l+p_l+n_l+q_l \geq 0 }}
\prod_{l=1}^L \left\{ \binom{ m_l + p_l}{ p_l} \binom{ n_l + q_l}{ q_l } \right\}
\nonumber \\
& \qquad \times F_0(r + \tilde{r}_0)
\langle \Omega , \prod_{l=1}^{L-1} \left\{
\underline{W}_{p_l,q_l}^{m_l,n_l}[w]( r + r_l ; K_l^{(m_l,n_l)}) F_{l}(H_f + r + \widetilde{r}_{l}) \right\} \nonumber \\
& \qquad \times
\underline{W}_{p_L,q_L}^{m_L,n_L}[w]( r + r_L ; K_L^{(m_L,n_L)}) \Omega \rangle F_L(r + \widetilde{r}_L) , \label{eq:complicated}
\end{align}
with
\begin{align}
& K^{(M,N)} := (K_1^{(m_1,n_1)}, \ldots , K_L^{(m_L,n_L)}) , \quad K_l^{(m_l,n_l)} := (k_l^{(m_l)},\widetilde{k}_l^{(n_l)}) , \label{eq:KMNdef}
\\
& r_l := \Sigma[\widetilde{k}_1^{(n_1)}] + \cdots + \Sigma[\widetilde{k}_{l-1}^{(n_{l-1})}] + \Sigma[{k}_{l+1}^{(m_{l+1})}] + \cdots + \Sigma[{k}_L^{(m_L)}] , \label{eq:rldef}
\\
& \widetilde{r}_l := \Sigma[\widetilde{k}_1^{(n_1)}] + \cdots + \Sigma[\widetilde{k}_{l}^{(n_{l})}] + \Sigma[{k}_{l+1}^{(m_{l+1})}] + \cdots + \Sigma[{k}_L^{(m_L)}] . \label{eq:rltildedef}
\end{align}
In the case \eqref{eq:WforFirst} we set
$$
	\underline{W}_{p_l,q_l}^{m_l,n_l}[w](\,\cdot\,;K_l^{(m_l,n_l)})
		= {{W}}_{\quad p_l,q_l}^{^{(0)}m_l,n_l}[w](K_l^{(m_l,n_l)}) \, ,
$$
which was defined in \eqref{eq:defofWW1} and in case \eqref{eq:WforSecond} we use
\eqref{eq:defofWW2}, i.e.
$$
	\underline{W}_{p_l,q_l}^{m_l,n_l}[w](r;K_l^{(m_l,n_l)})
		= W_{p_l,q_l}^{m_l,n_l}[w](r,K_l^{(m_l,n_l)}) \, .
$$
\end{theorem}
A proof can be found in \cite{BacFroSig98-2}. We note that the proof is essentially the same as the proof of Theorem 3.6 in \cite{BCFS} or Theorem 7.2 in
\cite{HasHer11-1}.

\bibliography{renormfordegen}

\end{document}